\documentclass[%
reprint,
superscriptaddress,
frontmatterverbose, 
preprintnumbers,
longbibliography,
 amsmath,amssymb,
 aps,
 pra,
notitlepage,
nofootinbib,
twocolumn
]{revtex4-1}

\usepackage{lipsum}

\usepackage{graphicx}
\usepackage{dcolumn}
\usepackage{bm}
\usepackage{hyperref}
\usepackage{subfigure}
\hypersetup{
colorlinks=true,
linkcolor=blue,
filecolor=blue,
citecolor=blue,  
urlcolor=blue,
}

\newcommand{\mycomment}[1]{}

\usepackage{comment}

\usepackage{thmtools}
\usepackage{thm-restate}

\usepackage{enumerate} 

\usepackage{amssymb}
\usepackage{mathtools}

\usepackage{mathrsfs}
\usepackage{multirow}
\usepackage{bbm}

\usepackage{xcolor}

\definecolor{sanddune}{rgb}{0.59, 0.44, 0.09}
\definecolor{darkblue}{RGB}{0,0,102}
\definecolor{darkred}{rgb}{0.5,0.,0.}
\definecolor{BlueViolet}{RGB}{138,43,226}
\definecolor{SkyBlue}{RGB}{30,144,255}
\definecolor{DarkGreen}{RGB}{0,100,0}

\usepackage{amsthm}
\usepackage{amsmath}
\theoremstyle{plain}
\newtheorem{thm}{Theorem}
\newtheorem{lem}[thm]{Lemma}
\newtheorem{prop}[thm]{Proposition}

\theoremstyle{definition}
\newtheorem{defn}{Definition}

\newcommand{\ket}[1]{|#1\rangle}
\newcommand{\bra}[1]{\langle #1|}

\newcommand{\ketbra}[2]{|#1\rangle\langle #2|}

\newcommand{\abs}[1]{\left|#1\right|}

\newcommand{\mc}{\mathcal}

\newcommand{\mbb}{\mathbb}

\newcommand{\stab}{\mathrm{STAB}}

\newcommand{\ba}{\begin{eqnarray}}
\newcommand{\ea}{\end{eqnarray}}

\DeclareMathOperator{\Tr}{Tr}

\newcommand{\z}{\mathbf{z}}
\newcommand{\x}{\mathbf{x}}
\newcommand{\y}{\mathbf{y}}
\newcommand{\s}{\mathbf{s}}

\usepackage{makecell}
\usepackage{multirow}

\begin{document}
\newcommand{\onenorm}[1]{\left\| #1 \right\|_1}
\newcommand{\twonorm}[1]{\left\| #1 \right\|_2}
\newcommand{\norm}[1]{\left\| #1 \right\|}

\newcommand{\ols}[1]{\mskip.5\thinmuskip\overline{\mskip-.5\thinmuskip {#1} \mskip-.5\thinmuskip}\mskip.5\thinmuskip} 
\newcommand{\olsi}[1]{\,\overline{\!{#1}}} 

\title{Noise robustness and threshold of many-body quantum magic}


\author{Fuchuan Wei}
\affiliation{Yau Mathematical Sciences Center, Tsinghua University, Beijing 100084, China}
\affiliation{Department of Mathematics, Tsinghua University, Beijing 100084, China}

\author{Zi-Wen Liu}
\affiliation{Yau Mathematical Sciences Center, Tsinghua University, Beijing 100084, China}

\date{\today}

\begin{abstract}
Understanding quantum magic (i.e., nonstabilizerness) in many-body quantum systems is challenging but essential to the study of quantum computation and many-body physics.
We investigate how noise affects magic properties in entangled many-body quantum states by quantitatively examining the magic decay under noise, with a primary aim being to understand the stability of magic associated with different kinds of entanglement structures.
As a standard model, we study hypergraph states, a representative class of many-body magic states, subject to depolarizing noise. First, we show that interactions facilitated by high-degree gates are fragile to noise. 
In particular, the $\mathrm{C}^{n-1}Z$ state family  exhibits a vanishing magic threshold of $\Theta(1/n)$.
Furthermore, we demonstrate efficiently preparable families of hypergraph states without local magic but with a non-vanishing magic threshold which signifies robust magic that is entirely embedded in global entanglement. We also discuss the qudit case based on the discrete Wigner formalism. 
\end{abstract}

\pacs{}
\maketitle


\section{Introduction}

One of the central quests in quantum information research is to understand why quantum computing is potentially more powerful than classical computing.  This has driven the development of several areas that have garnered major interest in the field, including quantum computational supremacy~\cite{Harrow2017supremacy} and quantum resource theory~\cite{Chitambar2019resource}, but a complete understanding is yet to be achieved.
Building on the Gottesman--Knill theorem~\cite{gottesman1998heisenberg,NielsenChuang,AaronsonGottesman04} which indicates that  stabilizer quantum computation can be efficiently simulated on classical computers,
the theory of nonstabilizerness, aka ``magic'' \cite{Bravyi2005universal, Veitch2014resource,Howard2017application,Liu2022manybody}, emerges as a resource theory that addresses the source of quantum computational advantages.

A major obstacle to the realization of quantum computation is the inherent susceptibility of physical quantum systems to
noise and imperfections which are ubiquitous in reality. These effects can significantly undermine the resource features of quantum systems and hence diminish their power and utility for applications, rendering whether and when the feature of interest remain effective under noise effects a question of fundamental importance. 
This question underlies 
many important research areas, especially 
noisy intermediate-scale quantum (NISQ) technologies~\cite{Preskill2018quantumcomputingin,Chen2023complexity} and quantum error correction~\cite{shor95,gottesman1997,NielsenChuang}, and has been investigated for various specific features such as
computational supremacy \cite{Aharonov2023polynomial,shao2024Simulating,schuster2024polynomial,sun2024sudden} and  entanglement \cite{RevModPhys.81.865,Contreras2022survival,Miller2023robustnessofEnt}.

Here, we embark on the study of noise robustness of magic in entangled many-body quantum systems, with the overarching goal of understanding how noise effects impact the quantum computational resources in large systems with different entanglement structures, which generates important insights into the interplay between magic and entanglement, the design of circuits and architectures for quantum computation, and so on.

As a neat yet highly versatile model for entangled magic states~\cite{Liu2022manybody}, hypergraph states~\cite{Rossi2013HGstate,qu2013encoding}, which have found importance in e.g.~many-body physics~\cite{Levin2012Braiding,Miller2016hierarchy,PhysRevLett.120.170503,Ellison2021symmetryprotected} and measurement-based quantum computing (MBQC)~\cite{Raussendorf2001OneWay,Miller2016hierarchy}, provide an apt playground for concretely investigating the relation between entanglement structures and magic properties.
The main approach of this work is to quantitatively analyze the decay of magic in hypergraph states under independent noise, which we consider as a standard model, mainly based on the well-behaved robustness of magic (RoM) \cite{Howard2017application} measure of mixed state magic and its bounds. 
The conceptual picture is that noise drives pure states to mixed ones, and as it intensifies, the magic gradually decays and eventually vanishes  at a certain point as the state is brought inside the stabilizer hull.
Particularly, we are interested in the  \textit{magic noise threshold}, defined as the noise rate at which the magic vanishes.
Our first message is the fragility of  magic provided by high-degree hyperedges in hypergraph states. Specifically, we show that  the magic provided by an edge decreases exponentially with the edge degree for fixed noise rates, and moreover, the magic threshold for the hypergraph state generated by a degree-$n$ edge (representing ${\mathrm{C}^{n-1}Z}$) scales as $\Theta(1/n)$.
Secondly, we look into the relationship between local and global magic
through the lens of magic noise threshold, uncovering an intriguing separation between them.
We show that 
while non-vanishing local magic naturally implies a non-vanishing magic threshold, the converse is not necessarily true, and remarkably, the hypergraph state model enables explicit and efficient constructions of counterexamples, namely, state families that are proved to host extensive and robust magic completely ``hidden'' in global entanglement. This is expected to have further applications and sheds light on the nature of many-body magic and its interplay with entanglement and complexity.
Additionally, we present a preliminary analysis of the qudit case with odd prime local dimensions using the discrete Wigner formalism, and point out various future directions. 

\section{Preliminaries}

\subsection{Hypergraph states}
Let $\mathrm{C}^{n-1}Z=\operatorname{diag}(1,\cdots,1,-1)$ denote the multi-controlled-$Z$ gate on $n$-qubits, with $\mathrm{C}^0Z=Z$.
Given a hypergraph $G=\{[n], E\}$, where $[n]:=\{1,\cdots,n\}$ is the set of vertices and $E\subset 2^{[n]}$ is the set of hyperedges, an associated hypergraph state is defined by
\begin{equation}
\ket{\Psi_G}=\prod_{e\in E}\mathrm{C}^{\abs{e}-1}Z_e\ket{+^n},
\end{equation}
where $\mathrm{C}^{\abs{e}-1}Z_e$ is the multi-controlled-$Z$ gate applied to the qubits in $e\subset{[n]}$. The number of vertices contained in an edge $e$ is referred to as the \emph{degree} of the edge. 
The characteristic function $f_{\Psi}:\mbb{Z}_2^n\rightarrow\mbb{Z}_2$ for $\Psi_G$ is defined as
$f_\Psi(x)=\sum_{e\in E}\prod_{i\in e}x_i$,
and we have $\ket{\Psi_G}=2^{-n/2}\sum_{\\s\in\mbb{Z}_2^n}(-1)^{f(\s)}\ket{\s}$.
We write $\ket{\mathrm{C}^{n-1}Z}=\mathrm{C}^{n-1}Z\ket{+^n}$.

\subsection{Robustness of magic}

Let $\stab_n$ denote the set of all $n$-qubit stabilizer states, namely the convex hull of all pure stabilizer states. Denote $\stab=\cup_{n\ge1}\stab_n$.
For a $n$-qubit state $\rho$, its robustness of magic (RoM), denoted by $\mc{R}(\rho)$, is defined as 
\begin{equation}\label{eq:RoM_primal}
\mathcal{R}(\rho)=\min _{\sigma, \tau \in \stab_n}\left\{2 a+1 \mid \rho=(a+1) \sigma-a \tau, a \geq 0\right\}.
\end{equation}
We summarize some key properties of RoM here: i) $\mc{R}(\cdot)$ is faithful, i.e., $\mc{R}(\sigma)=1$ iff $\sigma\in\stab$; ii) for all trace-preserving stabilizer channels $\mc{E}$, we have $\mc{R}\left(\mc{E}(\rho)\right)\le\mc{R}(\rho)$; iii)
$\mc{R}(\sigma\otimes\rho)=\mc{R}(\rho)$ for $\sigma\in\operatorname{STAB}$; iv) for  a set of states $\{\rho_k\}$ and a set of real numbers $\{p_k\}$  satisfying $\sum_kp_k=1$, the \emph{convexity} of $\mc{R}(\cdot)$ implies that $\mc{R}\left(\sum_kp_k\rho_k\right)\le\sum_k\abs{p_k}\mc{R}(\rho_k)$. 
Importantly, the RoM measure provides an operational charaterization of the hardness of classically simulating a quantum system: there exists a classical algorithm based on quasiprobability sampling that simulate Clifford quantum computation assisted by magic ancilla $\rho$ with overhead that scales as  $\mc{R}(\rho)^2$~\cite{Howard2017application}.

\section{Noise Robustness and Threshold of magic}

As a standard noise model, we primarily consider the independent depolarizing noise $\mc{E}_\lambda^{\otimes n}$ acting on the $n$-body quantum system, which leaves the qubits it acts on unchanged with probability $1-\lambda$, and replaces them with $\mbb{I}_2/2$ with probability $\lambda$. This noise channel can be expressed as $\mc{E}_{\lambda}^{\otimes n}=((1-\lambda)\mc{I}+\lambda\mc{G})^{\otimes n}$, where $\mc{I}$ is the identity channel and $\mc{G}(\sigma)=\Tr(\sigma)\mbb{I}_2/2$. 
When a qubit is subject to $\mc{G}$, it is decoupled with the rest of the many-body system and can be effectively traced out.
For an $n$-qubit hypergraph state $\Psi:=\ketbra{\Psi}{\Psi}$, and a subset of qubits $I \subset [n]$, consider tracing out the qubits in $I$. We have
\begin{equation}\label{eq:rule_of_partialtrace}
\Tr_{I}(\Psi)=\frac{1}{2^{\abs{I}}}\sum_{\mathbf{b}\in\mbb{Z}_2^{\abs{I}}}\Psi^{(I,\mathbf{b})},
\end{equation}
where $\Psi^{(I,\mathbf{b})}$ are $(n-\abs{I})$-qubit hypergraph states obtained by removing vertices and edges from $\Psi$. Detailed descriptions of the rule can be found in Appendix~\ref{app:proof_of_HGstate_PartialTrace}.
Note that the analysis in this paper is easily generalizable to any replacement noise channels with $\mc{G}(\sigma)=\Tr(\sigma)\rho_{\mathrm{g}}$, where $\rho_{\mathrm{g}}$ is a fixed stabilizer state.

\begin{figure}[t]
\centering
\includegraphics[width=0.49\textwidth]{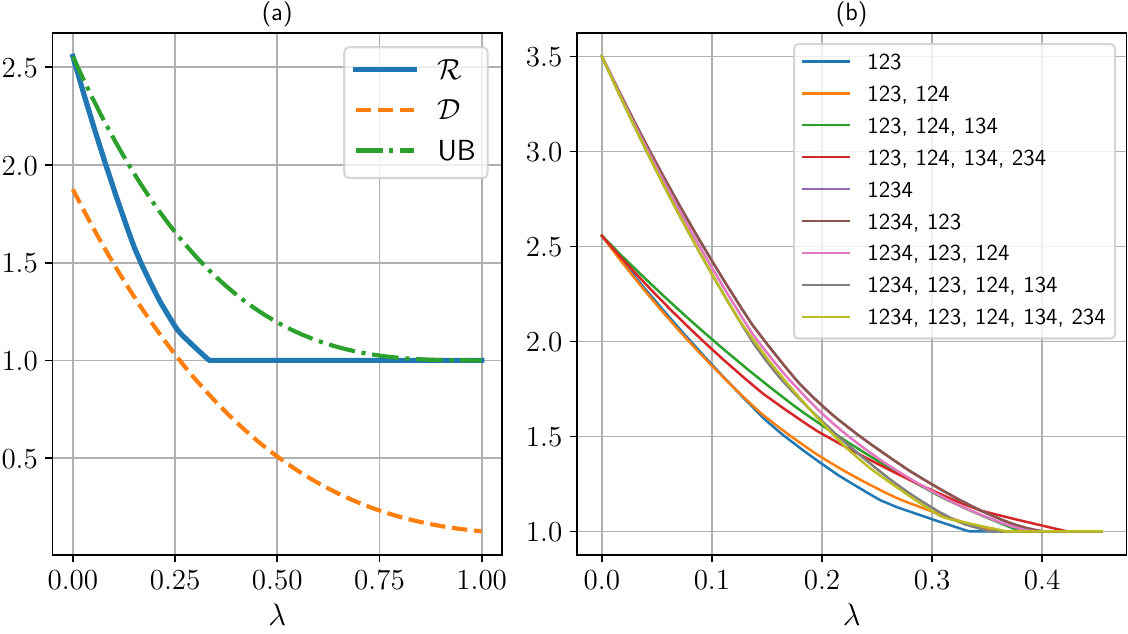}
\caption{
(a) Magic for $\ket{\mathrm{CC}Z}$ under $\mc{E}_\lambda^{\otimes 3}$ and its upper and lower bounds. The solid blue line depicts $\mc{R}\left(\mc{E}_\lambda^{\otimes 3}(\ketbra{\mathrm{CC}Z}{\mathrm{CC}Z})\right)$. The dashed green line is the upper bound given by Eq.~\eqref{eq:UB}, and the dashed orange line is the lower bound given by Eq.~\eqref{eq:LB}.
(b) $\mc{R}\left(\mc{E}_\lambda^{\otimes 4}(\Psi)\right)$ for all possible $4$-qubit hypergraph state $\Psi$ with edge degree $\ge3$, modulo permutations of qubits. 
The hypergraph state $\{1,2,3,4\}$ (labeled $1234$ in the figure legend) shares the same line as the hypergraph state $\{\{1,2,3,4\},\{1,2,3\}\}$.}
\label{fig:plot_R}
\end{figure}

Of particular interest in this work is the decay profile of magic, as characterized by $\mc{R}\left(\mc{E}_{\lambda}^{\otimes n}(\rho)\right)$, and the magic noise threshold above which the magic is eliminated.
\begin{defn}
For a state $\rho$ and $\epsilon\ge0$, we call $\lambda^*_\epsilon(\rho):=\inf_{\mc{R}\left(\mc{E}_{\lambda}^{\otimes n}(\rho)\right)\le1+\epsilon}\lambda$ the \emph{$\epsilon$-magic noise threshold}. 
For a family of states $\{\rho_n\}$ where $\rho_n$ is an $n$-qubit state, if for a fixed $\epsilon\ge0$ we have $\underset{n\rightarrow\infty}{\liminf}~\lambda^*_\epsilon(\rho_n)>0$,
we say $\{\rho_n\}$ has a \emph{non-vanishing magic threshold}.
\end{defn}

As a basic example, consider the hypergraph state $\Phi=\ketbra{\mathrm{CC}Z}{\mathrm{CC}Z}$.
By solving the dual program~\cite{Howard2017application}, we obtain the analytic expression as {$\mc{R}(\mc{E}_{\lambda}^{\otimes n}(\Phi))=\max_{1\le j\le9}\alpha_j(\lambda)$, where $\alpha_j(\lambda)=\Tr\left(\mc{E}_{\lambda}^{\otimes 3}(\Phi)A_j\right)$ are polynomials in $\lambda$ with degree at most $3$, with $A_j$ listed in Appendix~\ref{app:CCZ}.
We plot $\mc{R}(\mc{E}_{\lambda}^{\otimes n}(\Phi))$ by solid blue line in Fig.~\ref{fig:plot_R}(a).
By this analytical expression, we find that $\ket{\mathrm{CC}Z}$ has noiseless RoM $\mc{R}(\Phi)=23/9$, and magic threshold $\lambda^*_0=1/3$.

The computation of RoM is highly challenging, requiring solving a linear programming problem of size $2^{\Theta(n^2)}$~\cite{Howard2017application,Heinrich2019robustnessofmagic,hamaguchi2023handbook}, which is intractable for $n>5$ in general.
Nevertheless, we can write down simple and useful upper and lower bounds for it. 

By expressing $\mc{E}_{\lambda}^{\otimes n}(\cdot)=((1-\lambda)\mc{I}(\cdot)+\lambda\Tr(\cdot)\frac{\mbb{I}_2}{2})^{\otimes n}$ and using the convexity of $\mc{R}(\cdot)$, we find the following  upper bound:
\begin{equation}\label{eq:UB}
\mc{R}\left(\mc{E}_{\lambda}^{\otimes n}(\rho)\right)\le\sum_{I\subset[n]}(1-\lambda)^{n-\abs{I}}\lambda^{\abs{I}}\mc{R}\left(\Tr_{I}(\rho)\right),
\end{equation}
where the summation is taken over all subsets of $[n]$.
For example, for $\Phi=\ketbra{\mathrm{CC}Z}{\mathrm{CC}Z}$, $\mc{R}\left(\Tr_{I}(\Phi)\right)=1$ for all $I\neq\emptyset$, thus Eq.~\eqref{eq:UB} gives an upper bound $1+\frac{14}{9}(1-\lambda)^3$ (the green dashed line in Fig.~\ref{fig:plot_R}(a)). From this upper bound, we infer that $\mc{R}\left(\mc{E}_{\lambda}^{\otimes n}(\Phi)\right)$ gets below $1+\epsilon$ when $\lambda\ge1-(\frac{9}{14}\epsilon)^{\frac{1}{3}}$, yielding the bound  $\lambda^*_\epsilon\le1-(\frac{9}{14}\epsilon)^{\frac{1}{3}}$.
Since $\mc{R}\left(\Tr_{I}(\rho)\right)=\mc{O}(2^{n-\abs{I}})$ \cite{Liu2022manybody}, using Eq.~\eqref{eq:UB} we obtain 
$\mc{R}\left(\mc{E}_\lambda^{\otimes n}(\rho)\right)=\mc{O}\big((2-\lambda)^n\big)$  (see Appendix~\ref{app:generalUB}),
which holds for arbitrary $n$-qubit states $\rho$.

On the other hand, a lower bound of $\mc{R}(\rho)$ is given by
$\mc{D}(\rho):=\frac{1}{2^n}\sum_{P\in\mathrm{P}_n^+}\abs{\Tr(P\rho)}$~\cite{Howard2017application, Lorenzo2022stabilizer}, where $\mathrm{P}_n^+$ is the set of $n$-qubit Pauli operators with phase $+1$.
Since
$\Tr\left(P\mc{E}_{\lambda}^{\otimes n}(\rho)\right)=(1-\lambda)^{\operatorname{wt}(P)}\Tr(P\rho)$,
we obtain the following lower bound for $\mc{R}\left(\mc{E}_{\lambda}^{\otimes n}(\rho)\right)$ in terms of $\mc{D}\left(\mc{E}_{\lambda}^{\otimes n}(\rho)\right)$:
\begin{equation}\label{eq:LB}
\mc{R}\left(\mc{E}_{\lambda}^{\otimes n}(\rho)\right)\ge\frac{1}{2^n}\sum_{P\in\mathrm{P}_n^+}(1-\lambda)^{\operatorname{wt}(P)}\abs{\Tr(P\rho)},
\end{equation}
where $\operatorname{wt}(P)$ is the number of non-identity elements in $P$. For $\Phi=\ketbra{\mathrm{CC}Z}{\mathrm{CC}Z}$, we have $\mc{D}\left(\mc{E}_{\lambda}^{\otimes n}(\Phi)\right)=\frac{1}{16}(30-66\lambda+51\lambda^2-13\lambda^3)$ (the orange dashed line in Fig.~\ref{fig:plot_R}(a)). The root $\lambda\approx0.26$ of $\mc{D}\left(\mc{E}_{\lambda}^{\otimes n}(\Phi)\right)-1$ provides a lower bound for the threshold $\lambda^*_0$.

We also showcase in Fig.~\ref{fig:plot_R}(b) the numerical results for 4-qubit hypergraph states with different hyperedge structures. Depending on the presence of the degree-4 edge $\{1,2,3,4\}$, these hypergraphs exhibit noiseless RoM of $3.5$ or $23/9$.
A key observation is that increasing the number of edges results in a higher magic threshold $\lambda^*_0$ for the four hypergraphs lacking the degree-4 edge. In particular, the 3-complete hypergraph state $\{\{1,2,3\},\{1,2,4\},\{1,3,4\},\{2,3,4\}\}$ exhibits the highest magic threshold among all cases.
Also note that further adding the degree-4 edge degrades the magic threshold. Counterintuitively, the decay profile does not necessarily align with the threshold order and crossovers can occur (see e.g.~the green and red lines).

We analytically study the noise threshold and robustness properties of various representative hypergraph state families, with results summarized in Table~\ref{tab:threshold}.

\begin{table*}[htbp]
\centering
\resizebox{1\textwidth}{!}{
\begin{tabular}{c|c|c|c|c|c}
\hline
\hline
Hypergraph state $\Psi$ & Edge degree & Threshold & $\mc{R}(\mc{E}_{\lambda}^{\otimes n}(\Psi))$ upper bound & $\mc{R}(\mc{E}_{\lambda}^{\otimes n}(\Psi))$ lower bound & Local magic\\
\hline
\hline
$\ket{\mathrm{CC}Z}$ & \multirow{3}{*}{$3$} & $\lambda^*_0=1/3$ & \multicolumn{2}{c|}{$\mc{R}(\mc{E}_{\lambda}^{\otimes n}(\Psi))=\max_{1\le j\le9}\alpha_j(\lambda)$ [Fig.~\ref{fig:CCZ_R_exact}] } & $23/9$\\
\cline{1-1}\cline{3-6}
Union Jack lattice &  & $\lambda^*_0\ge\text{constant}>0$ & $\mc{O}\big((2-\lambda)^n\big)$ & $-$ & $1.0078$\\
\cline{1-1}\cline{3-6}
$3$-complete hypergraph &  & $0.39\le\lambda^*_\epsilon\le0.78$ & $1+2^{n+1}\big(1-(1-2^{-3/2})\lambda\big)^n$ & $\approx2^{(n-3)/2}(1-3\lambda/4)^n$ & $1+\mc{O}(2^{-n/2})$\\
\hline
$4$-complete hypergraph &$4$ & $\lambda^*_0\ge\text{constant}>0$ & $\mc{O}\big((2-\lambda)^n\big)$ & $\mc{D}\left(\mc{E}_{\lambda}^{\otimes n}(\Psi)\right)$ [Fig.~\ref{fig:plot_D_4complete}] &$1.25+o(1)$\\
\hline
High-degree hypergraphs & $\ge m$ & $\lambda^*_\epsilon=\mc{O}(1/m)$ & $1+\operatorname{poly}(m)(1-\lambda/2)^m$ & $-$ & $1+o(1)$\\
\hline
$\ket{\mathrm{C}^{n-1}Z}$ &$n$ & $\lambda^*_\epsilon=\Theta(1/n)$ & $1+4(1-\lambda/2)^n$ & $\approx 2 \left(1 - 3\lambda/4\right)^n$ & $1+\Theta\left((1/2)^n\right)$\\
\hline
Qudit $\ket{\mathrm{C}^{n-1}Z}$ & $n$ & $\lambda^*_0\ge0.42$ $(d=3)$ & $1+4M_d(d-1)^n\left(1-\frac{d-1}{d}\lambda\right)^n$ & $1+2\operatorname{sn}(\mc{E}_{\lambda}^{\otimes n}(\Psi))$ [Fig.~\ref{fig:plot_N_qudit_sym_d3_d5}] & $1+\mc{O}\left((\frac{d-1}{d})^n\right)$ \\
\hline
\hline
\end{tabular}
}
\caption{Noise robustness and threshold of the magic of various hypergraph states. The state $\ket{\mathrm{CC}Z}$ has 3 qubits, while other hypergraph states in the table have $n$ qubits/qudits. The qudit $\ket{\mathrm{C}^{n-1}Z}$ in the last row has local dimension $d$, whereas all other hypergraph states have local dimension 2. The threshold estimates are valid when $n$ is sufficiently large and for any fixed $\epsilon > 0$. 
The expressions $\alpha_j(\lambda)$ are polynomials in $\lambda$ of degree at most 3, with their definitions provided in Appendix~\ref{app:CCZ}. Hypergraph states with finite spatial dimension and local interaction, such as the Union Jack lattice, exhibit non-vanishing local magic, leading to a non-vanishing threshold. High-degree hypergraphs in the 5th row are composed of $\mc{O}(\log(m))$ arbitrarily placed edges each with degree $\ge m$; they are fragile to noise and have a vanishing threshold. For the upper bound in the last row, we define $M_d = \sup_{\rho \in D(\mbb{C}^d)} \mc{R}(\rho)$. Details for the technical calculations and numerical results are in Appendices~\ref{app:CCZ},~\ref{app:Example_CCCZ},~\ref{app:3-complete},~\ref{app:UnionJack},~\ref{app:4-complete},~\ref{app:qudit_CCCZ}.
}
\label{tab:threshold}
\end{table*}

\section{High-degree edges are fragile}\label{section:HighDegreeEdgeFragile}

We now demonstrate a general observation that the magic induced by high-degree edges is fragile. 
Specifically, under a constant noise rate $\lambda > 0$, the amount of magic contributed by these high-degree edges decreases exponentially with the edge size. Intuitively, this property arises because the magic provided by a $\mathrm{C}^{n-1}Z$ gate is upper bounded for all $n$, yet the magic spreads more extensively and is more exposed to noise as $n$ increases.
More generally, it is reasonable to expect that the magic induced by more extensive interactions become less robust against noise.

Define the magic capacity \cite{Seddon2019quantifying} of the $\mathrm{C}^{n-1}Z$ gate as $\mathcal{C}(\mathrm{C}^{n-1}Z)={\max}_{\ket{s}}\mathcal{R}\left(\mathrm{C}^{n-1}Z\otimes \mbb{I}_{2^n}\ket{s}\right)$, where the maximum is taken over all $2n$-qubit pure stabilizer states $\ket{s}$. It quantifies the maximum increase in RoM as the inequality $\mc{R}\left(\mathrm{C}^{n-1}Z\otimes\mbb{I}\rho\mathrm{C}^{n-1}Z\otimes\mbb{I}\right)\le\mc{C}(\mathrm{C}^{n-1}Z)\mc{R}(\rho)$ holds for {any $\rho$ on} an arbitrary number of qubits. 
Notice that when acting on the computational basis, $\mathrm{C}^{n-1}Z$ only flips the sign of $\ket{1^n}$, which induces a small perturbation and is expected to generate only a limited amount of magic.
Indeed, we prove in Appendix~\ref{app:MagicCapacity_CCCZ} that, for all $n$,
\begin{equation}\label{eq:MagicCapacity_CCCZ_UBis5}
\mathcal{C}\left(\mathrm{C}^{n-1}Z\right)<5.
\end{equation}
Although synthesizing $\mathrm{C}^{n-1}Z$ requires at least $n$ $T=\operatorname{diag}(1,e^{i\pi/4})$ gates~\cite{Beverland2020lower}, Eq.~\eqref{eq:MagicCapacity_CCCZ_UBis5} suggests that $\mathrm{C}^{n-1}Z$ can potentially be produced more efficiently using fewer magic resources.
Note that using similar techniques, we can upper bound the RoM of any $n$-qubit hypergraph state $\Psi$ by the second-order nonlinearity of its characteristic function. Specifically, we have $\mc{R}(\Psi)<1+4\chi\left(f_{\Psi}\right)$, where $\chi(f):=\min _{\text {quadratic } f^{\prime}} \operatorname{wt}\left(f+f^{\prime}\right)$ \cite{Liu2022manybody}.

Let $\Phi_n=\ketbra{\mathrm{C}^{n-1}Z}{\mathrm{C}^{n-1}Z}$ be the hypergraph state with a single edge containing all qubits.
By Eq.~\eqref{eq:MagicCapacity_CCCZ_UBis5}, we have $\mc{R}\left(\Phi_n\right)<5$ for all $n$. 
Since $\Phi_n$ is symmetric under permutations, its reduced density matrix depends only on the number of qubits removed.
By the rules described in Appendix~\ref{app:proof_of_HGstate_PartialTrace}, when $k$ qubits are removed, the $n$-edge is reduced to an $(n-k)$-edge only if the removed qubits are labeled by $\mathbf{b} = 1^{k}$; otherwise, the $n$-edge is eliminated. 
Therefore, $\Tr_{\overline{[k]}}(\Phi_n)=\left(1-\frac{1}{2^{k}}\right)\ketbra{+^{n-k}}{+^{n-k}}+\frac{1}{2^{k}}\Phi_{n-k}$, whose RoM is upper bounded by $\left(1-\frac{1}{2^{k}}\right)+\frac{5}{2^{k}}$ due to the convexity of $\mc{R}(\cdot)$. Using Eq.~\eqref{eq:UB}, we obtain
\begin{equation}\label{eq:CCCZ_UB}
\begin{aligned}
\mc{R}\left(\mc{E}_{\lambda}^{\otimes n}\left(\Phi_n\right)\right)
\le&1+4\left(1-\lambda/2\right)^n.
\end{aligned}
\end{equation}
which decays exponentially in $n$ for fixed $\lambda>0$.  
For any $\epsilon>0$,
we obtain $\lambda^*_\epsilon\le2\big(1-(\epsilon/4)^{\frac{1}{n}}\big)=\mc{O}(n^{-1})$.
By Eq.~\eqref{eq:LB}, we find a lower bound for $\mc{R}\left(\mc{E}_{\lambda}^{\otimes n}\left(\Phi_n\right)\right)$ of approximately $2 \left(1 - 3\lambda/4\right)^n$ (see Appendix~\ref{app:Example_CCCZ}), which implies $\lambda^*_\epsilon=\Omega(n^{-1})$.
To summarize, we have the following theorem:
\begin{thm}
For a fixed noise rate $0<\lambda<1$ and any $\epsilon>0$, the family $\{\ket{\mathrm{C}^{n-1}Z}\}$ has magic threshold $\lambda^*_\epsilon=\Theta(n^{-1})$.
\end{thm}

Notably, $\mathrm{C}^{n-1}Z$ gates are natural to implement on various physical platforms \cite{Kim2022iToffoli,Nguyen2024Floquet,Bluvstein2024Logical,2023arXiv231209111W}, suggesting the potential of considering $\mathrm{C}^{n-1}Z$ variants of magic state distillation and injection.
However, the fragility of $\Phi_n$ to noise implies inefficiency in distilling $\Phi_n$ from $\mc{E}_\lambda^{\otimes n}(\Phi_n)$ for large $n$. Since $\mc{R}\left(\mc{E}_\lambda^{\otimes n}(\Phi_n)^{\otimes K}\right)\le\left(1+4\left(1-\lambda/2\right)^n\right)^K$, at least $\Omega((1-\lambda/2)^{-n})$ copies of $\mc{E}_\lambda^{\otimes n}(\Phi_n)$ are required to achieve constant RoM, and therefore, to distill $\Phi_n$ with constant success probability \cite{Veitch2014resource,wang2020efficiently}.

Furthermore, in an arbitrary hypergraph state $\Psi$, the magic contributed by high-degree edges is fragile to noise. 
A general result goes as follows:
\begin{thm}\label{thm:HighDegreeEdgeFragile}
Let $\Psi$ be an $n$-qubit hypergraph state. Consider adding $K$ edges $e_1,\cdots,e_K\subset[n]$ to $\Psi$ to obtain a new state $\Phi$. Then $\mc{R}\left(\mc{E}_{\lambda}^{\otimes n}\left(\Phi\right)\right)$ is upper bounded by
\begin{equation}\label{eq:HighDegreeEdgeFragile}
\begin{aligned}
\mc{R}\left(\mc{E}_{\lambda}^{\otimes n}\left(\Psi\right)\right)+C_{\Psi}\sum_{\emptyset\neq J\subset[K]}(5^{\abs{J}}+1)\Big(1-\frac{\lambda}{2}\Big)^{\abs{\cup_{j\in J}e_j}},
\end{aligned}
\end{equation}
where $C_{\Psi}:=\underset{I\subset[n],\s\in\mbb{Z}_2^{\abs{I}}}{\max}~\mc{R}\big(\Psi^{(I,\mathbf{s})}\big)$ is a constant depending on $\Psi$.
\end{thm}
The detailed proof is given in Appendix~\ref{app:proof_of_HighDegreeEdgeFragile}.
Here we outline the key idea behind the proof. 
By Eq.~\eqref{eq:rule_of_partialtrace}, the difference $\mc{E}_{\lambda}^{\otimes n}(\Phi)-\mc{E}_{\lambda}^{\otimes n}(\Psi)$ can be expressed as $\sum_{I\subset[n]}(1-\lambda)^{n-\abs{I}}\lambda^{\abs{I}}\frac{1}{2^{\abs{I}}}\sum_{\s\in\mbb{Z}_2^{\abs{I}}}\left[\Phi^{(I,\s)}-\Psi^{(I,\s)}\right]\otimes\frac{\mbb{I}_I}{2^{\abs{I}}}$. Treating $(1-\lambda)^{n-\abs{I}}\lambda^{\abs{I}}\frac{1}{2^{\abs{I}}}$ as a joint probability distribution over $(I,\mathbf{s})$, we find that with probability at most $\left(1-\frac{\lambda}{2}\right)^{\abs{\cup_{j\in J}e_j}}$, $\Phi^{(I,\s)}$ differs from $\Psi^{(I,\s)}$ for edges in $J$. Using Eq.~\eqref{eq:MagicCapacity_CCCZ_UBis5} we know that $\mc{R}(\Phi^{(I,\s)})\le5^{\abs{J}}\mc{R}(\Psi^{(I,\s)})\le5^{\abs{J}}C_{\Psi}$. Summing over all possible $\emptyset\neq J\subset[K]$ and leveraging convexity of $\mc{R}(\cdot)$, we obtain Eq.~\eqref{eq:HighDegreeEdgeFragile}.

According to the exponent $\abs{\cup_{j\in J}e_j}$ in Eq.~\eqref{eq:HighDegreeEdgeFragile}, edges with higher degrees and less overlap are less robust against noise.
Consider the following two cases that illustrate the fragility of edges with degree $\ge m$:
\emph{Case~1}: The edges $\{e_j\}$ are disjoint, with $K=\mc{O}(\operatorname{poly}(m))$ and $\abs{e_j}\ge m$ for all $j$.
\emph{Case~2}: $\{e_j\}$ are placed arbitrarily, with $K=\mc{O}(\log(m))$ and $\abs{e_j}\ge m$ for all $j$.
In both cases, applying Theorem~\ref{thm:HighDegreeEdgeFragile}, we obtain the bound
\begin{equation}
\mc{R}\left(\mc{E}_{\lambda}^{\otimes n}\left(\Phi\right)\right)\le\mc{R}\left(\mc{E}_{\lambda}^{\otimes n}\left(\Psi\right)\right)+C_{\Psi}\operatorname{poly}(m)(1-\lambda/2)^{m}.
\end{equation}
Specifically, if we take $\Psi=\ketbra{+^n}{+^n}$ in case 2, then $\mc{R}\left(\mc{E}_{\lambda}^{\otimes n}\left(\Phi\right)\right)\le1+\operatorname{poly}(m)(1-\lambda/2)^m$, indicating that $\Phi$ has magic threshold $\lambda^*_\epsilon=\mc{O}(1/m)$.
We point out that with $K=2^m$ edges, the bound does not hold for Case 2 . 
Consider taking $\Psi=\ketbra{+^n}{+^n}$ and adding $2^m$ edges to $\Psi$ to obtain $\Phi$, each containing qubits $1,2,3$ and $m$ other qubits not included in any other edge.
As $m\rightarrow\infty$, $\Tr_{\overline{[3]}}(\Phi)$ approaches $(1-q)\ketbra{+^3}{+^3}+q\ketbra{\mathrm{CC}Z}{\mathrm{CC}Z}$ for $q=(1-e^{-2})/2$, which has RoM $1.6725$. Therefore, by Proposition~\ref{prop:LocalMagic_then_Threshold} below, we know $\Phi$ has a non-vanishing magic threshold.
There remains a significant gap between $\mc{O}(\operatorname{log}(m))$ and $2^m$, which we leave as a question.

\section{Local Magic and Magic Threshold}

For an $n$-qubit state $\rho$, define its maximum $K$-local RoM by $\mc{M}_K(\rho):=\max_{J\subset[n],\abs{J}=K}\mc{R}\left(\Tr_{\overline{J}}(\rho)\right)$, where $\overline{J}:=[n]-J$.
For a family of states $\{\rho_n\}$, if there exists a constant $K$ such that $\underset{n\rightarrow\infty}{\liminf}~\mc{M}_K(\rho_n)>1$ (or $\underset{n\rightarrow\infty}{\lim}~\mc{M}_K(\rho_n)=1$), we say $\{\rho_n\}$ has \emph{non-vanishing (or vanishing) local magic}.

\begin{prop}\label{prop:LocalMagic_then_Threshold}
Suppose $\{\rho_n\}$ is a family of states with non-vanishing local magic, then it has a non-vanishing magic threshold.
\end{prop}

The full proof is provided in Appendix~\ref{app:proof_LocalMagic_then_Threshold}. Here we sketch the main idea. Let $J_n$ ($\abs{J_n}=K$) denote the $K$ qubits on which $\rho_n$ has maximum RoM.
Since partial trace commutes with local noises, to ensure $\mc{E}_\lambda^{\otimes n}(\rho_n)\in\stab_n$, we must have $\mc{E}_\lambda^{\otimes K}\left(\Tr_{\overline{J_n}}\left(\rho_n\right)\right)\in\stab_K$. Since the $K$-qubit state $\Tr_{\overline{J_n}}\left(\rho_n\right)$ has RoM $\ge$ constant for all sufficiently large $n$, it can withstand constant rate of noises before falling into $\stab_K$.

However, the opposite direction of Proposition~\ref{prop:LocalMagic_then_Threshold} does not necessarily hold. In fact, we show that:
\begin{thm}\label{thm:3-complete-example}
There exist efficiently preparable families of many-body states with vanishing local magic yet a non-vanishing magic threshold.
\end{thm}

For example, consider the family of 3-complete hypergraph states $\{\Gamma_n\}$, where $\Gamma_n$ corresponds to the hypergraph with $n$ vertices and all possible edges of degree 3. Evidently, $\{\Gamma_n\}$ has gate complexity $\mc{O}(n^3)$.
For any constant $K$, the reduced density matrix of $\Gamma_n$ on $K$ qubits is a convex combination of four $(n-K)$-qubit hypergraph states, with all four coefficients approaching $1/4$ as $n \rightarrow \infty$, which can be proved to form a stabilizer state. Specifically, in Proposition~\ref{prop:3complete_localmagic_UB}, we show that $\mc{M}_K(\Gamma_n)\le1+2^{1+\frac{3}{2}K-\frac{n}{2}}$, which tends to 1 for all $K$ as $n \rightarrow \infty$, indicating that $\{\Gamma_n\}$ has vanishing local magic. Upper and lower bounds for $\mc{R}\left(\mc{E}_{\lambda}^{\otimes n}\left(\Gamma_n\right)\right)$  are derived in Appendix~\ref{app:3-complete} and summarized in Table~\ref{tab:threshold}, which leads to bounding the magic threshold for $\{\Gamma_n\}$ as $0.39 \leq \lambda^*_\epsilon \leq 0.78$, demonstrating a non-vanishing threshold.

Note that families of magic states generated by local interactions in finite spatial dimensions exhibit non-vanishing local magic and, by Proposition~\ref{prop:LocalMagic_then_Threshold}, non-vanishing magic thresholds. 
An example is the Union Jack lattice in 2D, which can support universal Pauli MBQC \cite{Miller2016hierarchy}.
On the other hand, local interactions with all-to-all connectivity can generate families of states with vanishing local magic yet non-vanishing magic thresholds.
$\{\Gamma_n\}$ represents such a family of many-body states that furthermore admits efficient preparation.
We stress that $\Gamma_n$ has a vanishing amount of magic even when restricted to $\Theta(n)$ qubits: dividing $\Gamma_n$ evenly among 4 parties results in a vanishing amount of magic for each party, as $\mc{M}_{n/4}(\Gamma_n)\le1+2^{1-n/8}\rightarrow1$.
However, the total magic remains substantial, as $\mc{R}(\Gamma_n)\ge\mc{D}(\Gamma_n)=\Theta(2^{n/2})$. That is, the ``non-local magic'' in these states is embedded in entanglement and can withstand a constant rate of noise. 
We anticipate that such states may find interesting applications through ``magic hiding'' or ``magic secret sharing''.

\section{Qudit cases: summary of results}\label{app:QuditManybodyStates}
We now briefly discuss the case of qudits, namely local dimension $d$, assuming $d$ is an odd prime number. 
The Wigner negativity \cite{Veitch2014resource} of an $n$-qudit state $\rho$ is defined as $\operatorname{sn}(\rho)=\sum_{\mathbf{u}:W_{\rho}(\mathbf{u})<0}\abs{W_{\rho}(\mathbf{u})}$ where $W_\rho(\mathbf{u})=\frac{1}{d^n} \operatorname{Tr}( A_{\mathbf{u}} \rho)$ is the discrete Wigner representation.
States with $\operatorname{sn}(\rho)=0$ are efficiently classically simulable \cite{Mari2012positiveWigner} and ineffective for magic state distillation. 
For all $\rho\in\mathrm{STAB}$ we have $\operatorname{sn}(\rho)=0$.
Additionally, we can set a lower bound for RoM in terms of Wigner negativity as $\mc{R}(\rho)\ge1+2\operatorname{sn}(\rho)$ \cite{Liu2022manybody}. 
More details for the definitions can be found in Appendix~\ref{app:DiscreteWignerfunction}.

Since $W_{\mc{E}_{\lambda}^{\otimes n}(\rho)}(\mathbf{u})=\frac{1}{d^n}\Tr\left(T_{\mathbf{u}}\mc{E}_{\lambda}^{\otimes n}\left(A_{\mathbf{0}}\right)T_{\mathbf{u}}^{\dagger}\rho\right)$, and $\mc{E}_{\lambda}^{\otimes n}\left(A_{\mathbf{0}}\right)$ becomes positive when $\lambda\ge\frac{d}{d+1}$, any $n$ qudit state can tolerate at most noise with intensity $\frac{d}{d+1}$ before becoming classical simulable. Furthermore, the single-qudit state $\ket{\psi}=\frac{1}{\sqrt{2}}(\ket{1}-\ket{d-1})$ maintains $\operatorname{sn}\left(\mc{E}_\lambda(\psi)\right)>0$ for all $\lambda<\frac{d}{d+1}$.
So for many-body states with local dimension $d$, the maximum magic threshold for Wigner negativity is $\frac{d}{d+1}$, which approaches $1$ as $d\rightarrow\infty$.
For the complete proof, see Appendix~\ref{app:maximumnoisethresholdforWignernegativity}.

An $n$-qudit hypergraph states \cite{Steinhoff2017QuditHypergraph,Xiong2018QuditHypergraph} (see Appendix~\ref{app:Qudithypergraphstates}) $\ket{\Psi}$ with characteristic function $f:\mbb{Z}_d^n\rightarrow\mbb{Z}_d$, can be written as $\ket{\Psi}=\frac{1}{\sqrt{d^n}}\sum_{\mathbf{s}\in\mbb{Z}_d^n}\omega^{f(\mathbf{s})}\ket{\mathbf{s}}$, where $\omega=e^{\frac{2\pi i}{d}}$.
Consider the $n$-qudit hypergraph state
\begin{equation}
\begin{aligned}
\ket{\mathrm{C}^{n-1}Z}:=\frac{1}{\sqrt{d^n}}\sum_{\s\in\mbb{Z}_d^n}\omega^{\prod_is_i}\ket{\s},
\end{aligned}
\end{equation}
and denote $\Phi_n=\ketbra{\mathrm{C}^{n-1}Z}{\mathrm{C}^{n-1}Z}$.
Unlike the qubit case, $\mc{R}(\Phi_n)$ is unbounded as $n$ increases, with its lower bound $1+2\operatorname{sn}(\Phi_n)$ growing rapidly with $n$ (Table~\ref{tab:qudit_CCCZ_sn}).
This indicates that a constant amount of noise is required to erase all its magic.
Analytically, we show in Appendix~\ref{app:quditCCCZLB} that for the qutrit $(d=3)$ case, for the family $\{\Phi_n\}_{n=3,7,11,\cdots}$ and $\mathbf{u}=(\mathbf{1},\mathbf{0})$, $W_{\mc{E}_\lambda^{\otimes n}(\Phi_n)}(\mathbf{u})$ remains negative for $\lambda\le0.42$ when $n$ is sufficiently large. 
This implies a magic threshold $\lambda^*_0 \ge 0.42$ for this family of states, in contrast to the threshold $\Theta(1/n)$ in the qubit case.
Numerical values of $1+2\operatorname{sn}\left(\mc{E}_\lambda^{\otimes n}(\Phi_n)\right)$ are shown in Fig.~\ref{fig:plot_N_qudit_sym_d3_d5}.
By applying Eq.~\eqref{eq:UB}, we derive in Appendix~\ref{app:quditCCCZUB} an upper bound for the noisy magic of $\Phi_n$ as
\begin{equation}
\mc{R}\left(\mc{E}_{\lambda}^{\otimes n}\left(\Phi_n\right)\right)\le1+4M_d(d-1)^n\big(1-\frac{d-1}{d}\lambda\big)^n,
\end{equation}
where $M_d := \sup_{\rho \in D(\mbb{C}^d)} \mc{R}(\rho)$. When $\lambda>\frac{d(d-2)}{(d-1)^2}$, this upper bound converges to $1$ exponentially with respect to $n$.

\begin{table}[htbp]
\centering
\resizebox{0.49\textwidth}{!}{
\begin{tabular}{|c|>{\centering\arraybackslash}m{0.13\linewidth} >{\centering\arraybackslash}m{0.13\linewidth} >{\centering\arraybackslash}m{0.13\linewidth} >{\centering\arraybackslash}m{0.13\linewidth} >{\centering\arraybackslash}m{0.13\linewidth}|}
\hline
$n$ & 3 & 4 & 5 & 6 & 7\\
\hline
$1+2\operatorname{sn}(\Phi_n)$ & 4.5555 & 5.7846 & 9.9368 & 12.7727 & 18.2254\\
\hline
\end{tabular}
}
\caption{The lower bound for $\mc{R}(\Phi_n)$, where $d=3$ and $\Phi_n=\ketbra{\mathrm{C}^{n-1}Z}{\mathrm{C}^{n-1}Z}$.}
\label{tab:qudit_CCCZ_sn}
\end{table}

\section{Discussion}

Motivated by the goal of understanding the impact of noise on quantum features,
we quantitatively studied the robustness of magic resource against standard noises in many-body quantum systems with versatile entanglement structures by leveraging the hypergraph state framework. Our findings, most notably the fragility of high-locality edges against noise and
efficiently constructible entangled states that provably possess ``purely global'' noise-robust magic, enrich the theory of many-body magic from both practical and mathematical perspectives.


We have some additional remarks.
Using similar techniques for proving Theorem~\ref{thm:HighDegreeEdgeFragile}, we can show that edges with degrees close to $n$ are particularly fragile under noise. As shown in Appendix~\ref{app:Fragility_of_bigbigEdges}, if $\Phi$ is obtained by adding arbitrary many edges with degrees $\ge n-c$ from $\Psi$, where $c$ is a constant, then we have
$\mc{R}\left(\mc{E}_{\lambda}^{\otimes n}\left(\Phi\right)\right)\le\mc{R}\left(\mc{E}_{\lambda}^{\otimes n}\left(\Psi\right)\right)+\operatorname{poly}(n)(1-\lambda/2)^n$.
Moreover, direct analysis of gates, particularly $\mathrm{C}^{n-1}Z$ gates which are commonly used in quantum algorithms, is compelling but more difficult.
As a preliminary result, we consider the dephasing noise $\mc{P}_\lambda$ with Kraus operators $E_0=\ketbra{0}{0}+\sqrt{1-\lambda}\ketbra{1}{1}$ and $E_1=\sqrt{\lambda}\ketbra{1}{1}$, and find by numerics that the vanishing point of the magic capacity \cite{Seddon2019quantifying} of $\mathrm{CC}Z$ composed with $\mc{P}^{\otimes 3}_\lambda$ occurs at $\lambda\approx0.645$.
We conjecture that the magic capacity threshold for $\mathrm{CC}Z$ under local depolarizing noise is $1/3$.

Various intriguing open questions remain. For instance, it would be compelling to identify the most robust hypergraph states. Moreover, the necessary and sufficient conditions for a family of states to have a non-vanishing magic threshold have yet to be established. It is also natural to expect connections between the noise robustness and MBQC power, which we leave for future study.
To conclude, our results enrich the field of many-body magic theory, and are anticipated to offer new perspectives and applications in the contexts of e.g.~quantum algorithm and circuit design, and quantum networks, which are worth further exploring as well. 

\begin{acknowledgments}
We thank Huihui Li, Yuguo Shao, Xin Wang, and Yuxuan Yan for insightful discussions. F.W. is supported by BMSTC and ACZSP (Grant No.~Z221100002722017). Z.-W.L.\ is supported in part by a startup funding from YMSC, Tsinghua University, and NSFC under Grant No.~12475023.
\end{acknowledgments}

\bibliography{MagicNoise}

\clearpage
\appendix
\onecolumngrid

\section{Rules of partial trace for hypergraph states}\label{app:proof_of_HGstate_PartialTrace}

\begin{lem}\label{lemma:HGstate_PartialTrace}
Let $\Psi$ be a $n$-qubit hypergraph state. For $I\subset [n]$ a subset of qubits, consider tracing out the qubits in $I$, we have
\begin{equation}
\Tr_{I}(\Psi)=\frac{1}{2^{\abs{I}}}\sum_{\mathbf{b}\in\mbb{Z}_2^{\abs{I}}}\Psi^{(I,\mathbf{b})},
\end{equation}
where $\Psi^{(I,\mathbf{b})}$ is the $(n-\abs{I})$-qubit hypergraph state obtained from $\Psi$ by \begin{enumerate}
    \item Remove the qubits labeled with ``0" along with the edges containing these qubits.
    \item Remove the qubits labeled with ``1".
    \item Keep the remaining edges with odd multiplicity, remove those with even multiplicity.
\end{enumerate}
\end{lem}
An example is depicted in Fig.~\ref{fig:HGstate_PartialTrace}.
A similar result holds for qudit hypergraphs, as shown in Lemma~\ref{lemma:HGstateQudit_PartialTrace}.
\begin{proof}
By definition of partial trace, we have
\begin{equation}\label{eq:A1}
\begin{aligned}
\Tr_{I}(\Psi)&=\sum_{\mathbf{b}\in\mbb{Z}_2^{\abs{I}}}\bra{\mathbf{b}}_{I}\otimes\mbb{I}_{\overline{I}}\Psi\ket{\mathbf{b}}_{I}\otimes\mbb{I}_{\overline{I}},\\
\end{aligned}
\end{equation}
where $\overline{I}=[n]-I$, $\ket{\mathbf{b}}_{I}$ is the computational basis state on $I$ and $\mbb{I}_{\overline{I}}$ be the identity operator on $\overline{I}$.
Without loss of generality, we assume $I$ be the first $\abs{I}$ qubits,
note that
\begin{equation}\label{eq:A2}
\begin{aligned}
\bra{\mathbf{b}}_{I}\otimes\mbb{I}_{\overline{I}}\ket{\Psi}=&\bra{\mathbf{b}}_{I}\otimes\mbb{I}_{\overline{I}}\frac{1}{\sqrt{2^n}}\sum_{\x\in\mbb{Z}_2^{n}}(-1)^{f_{\Psi}(\x)}\ket{\x}\\
=&\frac{1}{\sqrt{2^{\abs{I}}}}\frac{1}{\sqrt{2^{n-\abs{I}}}}\sum_{\y\in\mbb{Z}_2^{n-\abs{I}}}(-1)^{f_{\Psi}(\mathbf{b},\y)}\ket{\y}\\
=&\frac{1}{\sqrt{2^{\abs{I}}}}\frac{1}{\sqrt{2^{n-\abs{I}}}}\sum_{\y\in\mbb{Z}_2^{n-\abs{I}}}(-1)^{f_{\Psi^{(I,\mathbf{b})}}(\y)}\ket{\y}\\
=&\frac{1}{\sqrt{2^{\abs{I}}}}\ket{\Psi^{(I,\mathbf{b})}},
\end{aligned}
\end{equation}
where the third equality is given by 
\begin{equation}\label{eq:A3}
f_{\Psi}(\mathbf{b},\y)=f_{\Psi^{(I,\mathbf{b})}}(\y).
\end{equation}
We use an illustrative example to show the correctness of Eq.~\eqref{eq:A3}. For $\Psi$ shown in Fig.~\ref{fig:HGstate_PartialTrace}(a), we have $f_\Psi(\x)=x_1x_2x_3x_4+x_1x_2x_3+x_2x_3x_4+x_3x_4$, thus
\begin{equation}
\begin{aligned}
f_{\Psi}(0,y_2,y_3,y_4)&=y_2y_3y_4+y_3y_4=f_{\Psi^{(\{1\},0)}}(\y),\\
f_{\Psi}(1,y_2,y_3,y_4)&=y_2y_3y_4+y_2y_3+y_2y_3y_4+y_3y_4=y_2y_3+y_3y_4=f_{\Psi^{(\{1\},1)}}(\y),\\
\end{aligned}
\end{equation}
where $\y=(y_2,y_3,y_4)$. Combine Eq.~\eqref{eq:A1} and \eqref{eq:A2} we have
\begin{equation}
\Tr_{I}(\Psi)=\frac{1}{2^{\abs{I}}}\sum_{\mathbf{b}\in\mbb{Z}_2^{\abs{I}}}\Psi^{(I,\mathbf{b})}.
\end{equation}
\end{proof}

\begin{figure}[t]
\centering
\includegraphics[width=0.49\textwidth]{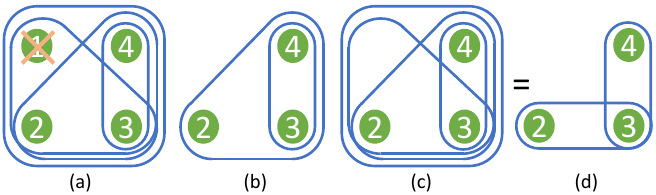}
\caption{Sketch of the rules for partial trace on hypergraph states. (a) We trace out qubit $1$ from the hypergraph state $\Psi$ with edges $\{1,2,3,4\},\{1,2,3\},\{2,3,4\},\{3,4\}$. (b) When qubit $1$ is labeled by $b=0$, remove qubit $1$ and the edges containing it. The resulting state $\Psi^{(\{1\},0)}$ has edges $\{2,3,4\},\{3,4\}$. (c) When qubit $1$ is labeled by $b=1$, remove the qubit $1$ while keeping the edges containing it. Remove the edges $\{2,3,4\}$ which appears with even multiplicity. (d) $\Psi^{(\{1\},1)}$ has edges $\{2,3\},\{3,4\}$. Finally, we can write $\Tr_{\{1\}}(\Psi)=\frac{1}{2}\Psi^{(\{1\},0)}+\frac{1}{2}\Psi^{(\{1\},1)}$.
}
\label{fig:HGstate_PartialTrace}
\end{figure}

\section{Example: $\ket{\mathrm{CC}Z}$}\label{app:CCZ}

For an $n$-qubit state $\rho$, its RoM can be computed by the following dual optimization problem:
\begin{equation}\label{eq:RoM_dual}
\begin{aligned}
\mc{R}(\rho)=\textbf{\text{max}}~ & \Tr(\rho A) ~~~~~~~~~~\text{ over Hermitian matrices }A,\\
\textbf{ s.t. } & \left|\Tr\left(\phi A\right)\right| \leq 1 ~~\text{ for all }\phi\in\stab_n.
\end{aligned}
\end{equation}
By solving this dual problem for $\rho=\mc{E}_{\lambda}^{\otimes 3}(\ketbra{\mathrm{CC}Z}{\mathrm{CC}Z})$ numerically, we can write the exact expression as $\mc{R}\left(\mc{E}_{\lambda}^{\otimes 3}(\ketbra{\mathrm{CC}Z}{\mathrm{CC}Z})\right)=\max_{1\le j\le9}\alpha_j(\lambda)$, as shown in Fig.~\ref{fig:CCZ_R_exact}. The Hermitian matrices $A_1\cdots,A_9$ are defined in Table~\ref{tab:Aj_definition}.

Particularly, we have $\alpha_1(0)=23/9\approx 2.55556$. Let $\alpha_8(\lambda^*)=1$, we get $\lambda^*=1/3$.

\begin{figure}[t]
\centering
\includegraphics[width=0.52\textwidth]{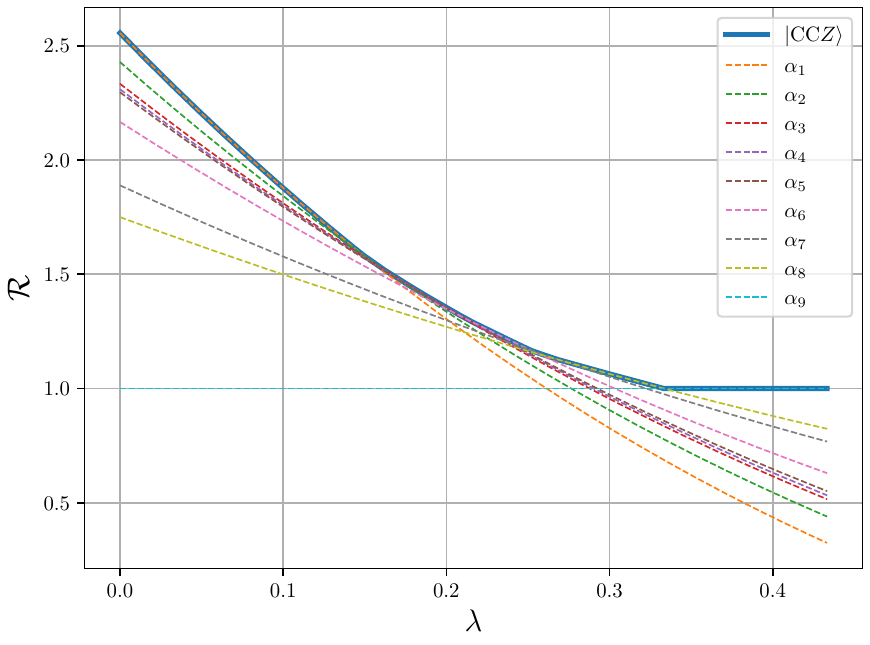}
\caption{The thick blue line represents $\mc{R}\left(\mc{E}_{\lambda}^{\otimes 3}(\ketbra{\mathrm{CC}Z}{\mathrm{CC}Z})\right)$.
The dashed lines represent $\alpha_j(\lambda)=\Tr\left(\mc{E}_{\lambda}^{\otimes 3}(\ketbra{\mathrm{CC}Z}{\mathrm{CC}Z})A_j\right)$ for $i=1,\cdots,9$. Each $A_j$ is a solution to the dual optimization problem \eqref{eq:RoM_dual} for the corresponding $\lambda$.}
\label{fig:CCZ_R_exact}
\end{figure}

\begin{table}[htbp]
\centering
\resizebox{0.78\textwidth}{!}{
\begin{tabular}{c|ccccccccc}
 & $A_1$ & $A_2$ & $A_3$ & $A_4$ & $A_5$ & $A_6$ & $A_7$ & $A_8$ & $A_9$ \\
\hline
000&$-0.55556$&$-0.42857$&$-0.33333$&$-0.30952$&$-0.2963$&$-0.16667$&$0.11111$&$0.25$&$1$\\
001&$0.22222$&$0.19048$&$0.14815$&$0.14286$&$0.22222$&$0.22222$&$0.2963$&$0.25$&$0$\\
002&$0$&$-0.09524$&$0.04345$&$0.05587$&$-0.01459$&$-0.03766$&$0.05978$&$0$&$0$\\
003&$0$&$0.04762$&$0.05556$&$0.04762$&$0.07407$&$0$&$0$&$0$&$0$\\
011&$0.22222$&$0.2381$&$0.25926$&$0.2619$&$0.22222$&$0.22222$&$0.07407$&$0$&$0$\\
012&$-0.03704$&$-0.07937$&$0.12453$&$0.00747$&$-0.03813$&$-0.07206$&$-0.07615$&$0$&$0$\\
013&$0.22222$&$0.28571$&$0.35185$&$0.35714$&$0.27778$&$0.27778$&$0.2037$&$0.25$&$0$\\
022&$0.22222$&$0.33333$&$0.37037$&$0.38095$&$0.40741$&$0.38889$&$0.2963$&$0.25$&$0$\\
023&$-0.07407$&$-0.01587$&$-0.01945$&$-0.1016$&$-0.01046$&$0.1797$&$-0.06416$&$-0.08333$&$0$\\
033&$0$&$0$&$0.03704$&$0.04762$&$-0.11111$&$-0.11111$&$-0.03704$&$0$&$0$\\
111&$0.22222$&$0.14286$&$0.11111$&$0.09524$&$-0.03704$&$-0.16667$&$-0.22222$&$0$&$0$\\
112&$0.07407$&$-0.04762$&$-0.02129$&$0.03337$&$-0.06508$&$0.08431$&$0.06854$&$0.08333$&$0$\\
113&$0.22222$&$0.09524$&$0.01852$&$0$&$0$&$-0.05556$&$-0.18519$&$-0.25$&$0$\\
122&$0.22222$&$0.14286$&$0.07407$&$0.07143$&$0.11111$&$0.11111$&$0.25926$&$0.25$&$0$\\
123&$-0.03704$&$0.04762$&$0.05658$&$-0.00796$&$0.07319$&$-0.04424$&$-0.01237$&$0$&$0$\\
133&$-0.22222$&$-0.09524$&$-0.03704$&$-0.02381$&$-0.03704$&$-0.05556$&$-0.07407$&$0$&$0$\\
222&$0$&$0.04762$&$0.09188$&$-0.00096$&$0.11785$&$0.11297$&$-0.17935$&$0$&$0$\\
223&$-0.22222$&$-0.19048$&$-0.12963$&$-0.11905$&$-0.18519$&$-0.11111$&$-0.03704$&$0$&$0$\\
233&$-0.07407$&$-0.09524$&$-0.02715$&$0.02347$&$-0.06316$&$-0.07244$&$0.16244$&$-0.08333$&$0$\\
333&$0$&$0.14286$&$0.16667$&$0.16667$&$0.07407$&$-0.16667$&$-0.55556$&$-0.75$&$0$\\
\end{tabular}
}
\caption{We give definition for $A_1,\cdots,A_9$ by providing their Pauli coefficients $\frac{1}{2^3}\Tr(A_jP)$, for all $P\in\mathrm{P}_3^+$. In the row labels, $0,1,2,3$ represents $\mbb{I},X,Y,Z$, respectively. Since $A_1,\cdots,A_9$ are taken to possess permutation symmetry, only a part of the coefficients are listed. For example, the element in row 013 and column $A_1$ is $0.22222$, which means $\frac{1}{2^3}\Tr(A_1P)=0.22222$ for $P=\mbb{I}\otimes X\otimes Z,\mbb{I}\otimes Z\otimes X,Z\otimes \mbb{I}\otimes X,Z\otimes X\otimes \mbb{I},X\otimes \mbb{I}\otimes Z$, and $X\otimes Z\otimes \mbb{I}$.}
\label{tab:Aj_definition}
\end{table}

\section{General upper bound for magic under noise}\label{app:generalUB}

For arbitrary $k$-qubit state $\sigma$ we have $\mc{R}(\sigma)\le\sqrt{2^{k}(2^{k}+1)}$ \cite{Liu2022manybody}, thus for arbitrary $n$-qubit state $\rho$, and arbitrary $0\le\lambda\le1$, we have
\begin{equation}
\begin{aligned}
\mc{R}\left(\mc{E}_\lambda^{\otimes n}(\rho)\right)\le&\sum_{J\subset[n]}(1-\lambda)^{\abs{J}}\lambda^{n-\abs{J}}\mc{R}\left(\Tr_{\overline{J}}(\rho)\right)\\
=&\sum_{k=0}^n\sum_{\abs{J}=k}(1-\lambda)^{\abs{J}}\lambda^{n-\abs{J}}\mc{R}\left(\Tr_{\overline{J}}(\rho)\right)\\
\le&\sum_{k=0}^n\binom{n}{k}(1-\lambda)^{k}\lambda^{n-k}\sqrt{2^{k}(2^{k}+1)}\\
\le&\sum_{k=0}^n\binom{n}{k}(1-\lambda)^{k}\lambda^{n-k}2^{k}(1+2^{-k-1})\\
=&(2-\lambda)^n+2^{-1}\left(\frac{1+\lambda}{2}\right)^n\\
=&\mc{O}\big((2-\lambda)^n\big).
\end{aligned}
\end{equation}

\section{Magic capacity of $\mathrm{C}^{n-1}Z$ is bounded}\label{app:MagicCapacity_CCCZ}

\begin{lem}\label{lem:CCCZ_magic_UB}
$\mathcal{C}\left(\mathrm{C}^{n-1}Z\right)<5$ holds for all $n$.
\end{lem}

\begin{proof}
For quantum channels whose Kraus operators are diagonal in the computational basis, we do not need to append an extra system to obtain the maximum output robustness (\cite{Seddon2019quantifying}, Theorem A1). Specifically, for $\mathrm{C}^{n-1}Z$ we have 
\begin{equation}
\mathcal{C}\left(\mathrm{C}^{n-1}Z\right)=\underset{\ket{s} \in \mc{S}_{2n}}{\max}\mathcal{R}\left(\mathrm{C}^{n-1}Z\otimes \mbb{I}_{2^n}\ket{s}\right)=\underset{\ket{s} \in \mc{S}_{n}}{\max} \mathcal{R}\left(\mathrm{C}^{n-1}Z\ket{s}\right),
\end{equation}
where $\mc{S}_n$ denotes the set of $n$ qubit pure stabilizer states. Any pure $n$-qubit stabilizer state has the form
\begin{equation}
\ket{\mc{K},q,\mathbf{b}}:=\frac{1}{\sqrt{\abs{\mc{K}}}}\sum_{x\in\mc{K}}i^{\mathbf{b}\cdot \x}(-1)^{q(\x)}\ket{\x},
\end{equation}
where $\mc{K}\subset\mbb{F}_2^n$ is an affine subspace, $\mathbf{b}\in\mbb{F}_2^n$, $q$ is a quadratic form, and $i=\sqrt{-1}$. If $1^n\notin\mc{K}$, then 
\begin{equation}
\begin{aligned}
\mathrm{C}^{n-1}Z\ket{\mc{K},q,\mathbf{b}}=\frac{1}{\sqrt{\abs{\mc{K}}}}\sum_{\x\in\mc{K}}i^{\mathbf{b}\cdot \x}(-1)^{q(\x)+\prod_{i=1}^{n}x_i}\ket{\x}=\ket{\mc{K},q,\mathbf{b}},
\end{aligned}
\end{equation}
whose RoM $=1$.
Suppose $1^n\in\mc{K}$, we have
\begin{equation}\label{eq:CCCZ_on_STAB}
\begin{aligned}
&\mathrm{C}^{n-1}Z\ketbra{\mc{K},q,\mathbf{b}}{\mc{K},q,\mathbf{b}}\mathrm{C}^{n-1}Z-\ketbra{\mc{K},q,\mathbf{b}}{\mc{K},q,\mathbf{b}}\\
=&\frac{1}{\abs{\mc{K}}}\sum_{\x,\y\in\mc{K}}i^{\mathbf{b}\cdot \x-\mathbf{b}\cdot \y}(-1)^{q(\x)+q(\y)+\prod_{i=1}^{n}x_i+\prod_{i=1}^{n}y_i}\ketbra{\x}{\y}-\frac{1}{\abs{\mc{K}}}\sum_{\x,\y\in\mc{K}}i^{\mathbf{b}\cdot \x-\mathbf{b}\cdot \y}(-1)^{q(\x)+q(\y)}\ketbra{\x}{\y}\\
=&\frac{1}{\abs{\mc{K}}}\sum_{\x,\y\in\mc{K}}i^{\mathbf{b}\cdot \x-\mathbf{b}\cdot \y}(-1)^{q(\x)+q(\y)}\left[(-1)^{\prod_{i=1}^{n}x_i+\prod_{i=1}^{n}y_i}-1\right]\ketbra{\x}{\y}\\
=&\frac{1}{\abs{\mc{K}}}\sum_{\substack{\x,\y\in\mc{K},\\\x=1^n,\y\neq1^n}}i^{\mathbf{b}\cdot \x-\mathbf{b}\cdot \y}(-1)^{q(\x)+q(\y)}(-2)\ketbra{\x}{\y}+\frac{1}{\abs{\mc{K}}}\sum_{\substack{\x,\y\in\mc{K},\\\x\neq1^n,\y=1^n}}i^{\mathbf{b}\cdot \x-\mathbf{b}\cdot \y}(-1)^{q(\x)+q(\y)}(-2)\ketbra{\x}{\y}\\
=&-\frac{2}{\abs{\mc{K}}}\sum_{\substack{\y\in\mc{K},\y\neq1^n}}(-1)^{q(1^n)+q(\y)}\big(i^{\mathbf{b}\cdot 1^n-\mathbf{b}\cdot \y}\ketbra{1^n}{\y}+i^{\mathbf{b}\cdot \y-\mathbf{b}\cdot 1^n}\ketbra{\y}{1^n}\big).
\end{aligned}
\end{equation}

Note that $\mathbf{b}\cdot1^n$ and $\mathbf{b}\cdot \y$ takes value in $0,1$. When $i^{\mathbf{b}\cdot 1^n-\mathbf{b}\cdot \y}=1$, we can write
\begin{equation}
\ketbra{1^n}{\y}+\ketbra{\y}{1^n}=\ketbra{X^+_{\y1^n}}{X^+_{\y1^n}}-\ketbra{X^-_{\y1^n}}{X^-_{\y1^n}},
\end{equation}
a linear combination of pure stabilizer states with total absolute coefficient $2$, where $\ket{X^\pm_{\y1^n}}=\frac{1}{\sqrt{2}}(\ket{\y}\pm\ket{1^n})$. Note that $\ket{X^\pm_{\y1^n}}$ can be generated from $\ket{\operatorname{GHZ}_n}=\frac{1}{\sqrt{2}}(\ket{0^n}+\ket{1^n})$ by applying a Pauli $Z$ and some $\operatorname{CNOT}$'s, thus is a stabilizer state.

When $i^{\mathbf{b}\cdot 1^n-\mathbf{b}\cdot \y}=\pm i$, we can write
\begin{equation}
\pm\big(i\ketbra{1^n}{\y}-i\ketbra{\y}{1^n}\big)=\pm\big(\ketbra{Y^+_{\y1^n}}{Y^+_{\y1^n}}-\ketbra{Y^-_{\y1^n}}{Y^-_{\y1^n}}\big),
\end{equation}
a linear combination of pure stabilizer states with total absolute coefficient $2$, where $\ket{Y^\pm_{\y1^n}}=\frac{1}{\sqrt{2}}(\ket{\y}\pm i\ket{1^n})$. Note that $\ket{Y^\pm_{\y1^n}}$ can be generated from $\ket{\operatorname{GHZ}_n}=\frac{1}{\sqrt{2}}(\ket{0^n}+\ket{1^n})$ by applying a Pauli $Z$, a phase gate $S$, and some $\operatorname{CNOT}$'s, thus is a stabilizer state.

In conclusion, every term in the summand of the last equation of Eq.~\eqref{eq:CCCZ_on_STAB} can be written as a linear combination of pure stabilizer states with a total absolute coefficient $2$. By convexity of $\mc{R}(\cdot)$, we obtain
\begin{equation}
\begin{aligned}
&\mc{R}\left(\mathrm{C}^{n-1}Z\ket{\mc{K},q,\mathbf{b}}\right)\le\mc{R}(\ket{\mc{K},q,\mathbf{b}})+\frac{2}{\abs{\mc{K}}}\big(\abs{\mc{K}}-1\big)2<5.
\end{aligned}
\end{equation}
Taking maximum over all the pure stabilizer states, we obtain 
\begin{equation}
\mathcal{C}\left(\mathrm{C}^{n-1}Z\right)=\underset{\ket{s} \in \mc{S}_{n}}{\max} \mathcal{R}\left(\mathrm{C}^{n-1}Z\ket{s}\right)<5.
\end{equation}

\end{proof}

\section{Example: $\ket{\mathrm{C}^{n-1}Z}$}\label{app:Example_CCCZ}

\subsection{Lower bound}
In the following proofs, for $\mathbf{z}=(z_1,\cdots,z_n),\mathbf{x}=(x_1,\cdots,x_n)\in\mbb{Z}_2^n$, denote $P_{(\mathbf{z},\mathbf{x})}\in\mathrm{P}_n^+$ be the Pauli operator which equals to $\left(\otimes_{i=1}^nZ^{z_i}\right)\left(\otimes_{i=1}^nX^{x_i}\right)$ up to a global phase $e^{i\theta}$. We have
\begin{equation}
\begin{aligned}
P_{(\mathbf{z},\mathbf{x})}\ket{\Psi_f}=&e^{i\theta}\left(\otimes_{i=1}^nZ^{z_i}\right)\left(\otimes_{i=1}^nX^{x_i}\right)\frac{1}{\sqrt{2^n}}\sum_{\mathbf{s}\in\mbb{Z}_2^n}(-1)^{f(\mathbf{s})}\ket{\mathbf{s}}\\
=&e^{i\theta}\frac{1}{\sqrt{2^n}}\sum_{\mathbf{s}\in\mbb{Z}_2^n}(-1)^{f(\mathbf{s}+\x)+\z\cdot\s}\ket{\mathbf{s}},
\end{aligned}
\end{equation}
where $\cdot$ is the inner product in $\mbb{Z}_2^n$. Thus
\begin{equation}\label{eq:Pauli_spectrum_abs}
\abs{\bra{\Psi_f}P_{(\mathbf{z},\mathbf{x})}\ket{\Psi_f}}=\Big|\frac{1}{2^n} \sum_{\s\in\mbb{Z}_2^n}(-1)^{f(\s)+f(\s+\x)+\z \cdot \s}\Big|.
\end{equation}

\begin{lem}
For $\z,\x\in\mbb{Z}_2^n$, we have
\begin{equation}
\begin{aligned}
\abs{\bra{\mathrm{C}^{n-1}Z}P_{(\mathbf{z},\mathbf{x})}\ket{\mathrm{C}^{n-1}Z}}=
\begin{cases}
1, & \x=0^n\text{ and }\z=0^n,\\
0, & \x=0^n\text{ and }\z\neq0^n,\\
1-\frac{1}{2^{n-2}}, & \x\neq0^n\text{ and }\z=0^n,\\
\frac{1}{2^{n-2}}, & \x\neq0^n\text{ and }\z\neq0^n\text{ and }\z\cdot\x=0,\\
0, & \x\neq0^n\text{ and }\z\neq0^n\text{ and }\z\cdot\x=1.\\
\end{cases}
\end{aligned}
\end{equation}
\end{lem}
\begin{proof}
For $\ket{\mathrm{C}^{n-1}Z}$ we have $f(\s)=\prod_{i}s_i$.

When $\x=0^n$ and $\z\neq0^n$, $f(\s)+f(\s+\x)+\z\cdot\s=\z\cdot\s$. The number of $\s\in\mbb{Z}_2^n$ such that $\z\cdot\s=0$ or $\z\cdot\s=1$ all equals to $2^{n-1}$. Thus $\sum_{\s\in\mbb{Z}_2^n}(-1)^{\z\cdot\s}=0$.

When $\x\neq0^n$, $\sum_{\s\in\mbb{Z}_2^n}(-1)^{f(\s)+f(\s+\x)+\z\cdot\s}=\sum_{\s\in\mbb{Z}_2^n}(-1)^{\z\cdot\s}-2(-1)^{\z\cdot1}-2(-1)^{\z\cdot(\x+1^n)}$, which equals to $2^n-4$ when $\z=0^n$, and equals to $-2(-1)^{\z\cdot1}\left(1+(-1)^{\z\cdot\x}\right)$ when $\z\neq0^n$.
\end{proof}

Denote $\x\vee\z$ ($\x\land \z$) be the bit-wise ``or" (``and") of $\x$ and $\z$. Given the Pauli spectrum of $\ket{\mathrm{C}^{n-1}Z}$, we can use Eq.~\eqref{eq:LB} to provide a lower bound as
\begin{equation}\label{eq:CCCZ_noisy_LB}
\begin{aligned}
&\mc{D}\left(\mc{E}_\lambda^{\otimes n}\left(\ketbra{\mathrm{C}^{n-1}Z}{\mathrm{C}^{n-1}Z}\right)\right)\\
=&\frac{1}{2^n}\sum_{\x,\z}\abs{\Tr\left(P_{(\mathbf{z},\mathbf{x})}\mc{E}_\lambda^{\otimes n}\left(\ketbra{\mathrm{C}^{n-1}Z}{\mathrm{C}^{n-1}Z}\right)\right)}\\
=&\frac{1}{2^n}\sum_{\x,\z}(1-\lambda)^{\abs{\x\vee \z}}\abs{\Tr\left(P_{(\mathbf{z},\mathbf{x})}\ketbra{\mathrm{C}^{n-1}Z}{\mathrm{C}^{n-1}Z}\right)}\\
=&\frac{1}{2^n}\Bigg[1+\sum_{\substack{\x\neq0^n,\\\z=0^n}}(1-\lambda)^{\abs{\x}}\left(1-\frac{1}{2^{n-2}}\right)+\sum_{\substack{\x\neq0^n,\z\neq0^n,\\\abs{\x\land \z}\text{ even}}}(1-\lambda)^{\abs{\x\vee \z}}\frac{1}{2^{n-2}}\Bigg]\\
=&\frac{1}{2^n}\Bigg[1+\sum_{k=1}^n\binom{n}{k}(1-\lambda)^{k}\left(1-\frac{1}{2^{n-2}}\right)+\sum_{k=2}^n\binom{n}{k}\Big[(2^k-2)+\sum_{m=1}^{\lfloor k/2\rfloor}\binom{k}{2m}2^{k-2m}\Big](1-\lambda)^{k}\frac{1}{2^{n-2}}\Bigg]\\
=&\frac{1}{2^n}\Bigg[1+\left(1-\frac{1}{2^{n-2}}\right)\left((2-\lambda)^n-1\right)+\frac{1}{2^{n-2}}\sum_{k=2}^n\binom{n}{k}\Big[(2^k-2)+\Big(\frac{1}{2}(1+3^k)-2^k\Big)\Big](1-\lambda)^{k}\Bigg]\\
=&4^{-n} \left(8 + 2 (4 - 3 \lambda)^n + (4 - 2 \lambda)^n - 10 (2 - \lambda)^n\right).\\
\end{aligned}
\end{equation}
In the third term of the fifth line, $k$ represents $\abs{\x\vee \z}$, $2m$ represents $\abs{\x\land \z}$, $2^k-2$ represents all possible choices to assign for the $k-2m$ $1$'s to $\x$ and $\z$ when $2m=0$, $2^{k-2m}$ represents all possible choices to assign for the $k-2m$ $1$'s to $\x$ and $\z$ when $2m>0$.

By Eq.~\eqref{eq:CCCZ_noisy_LB}, we have
\begin{equation}
\begin{aligned}
\mc{D}\left(\mc{E}_\lambda^{\otimes n}\left(\ketbra{\mathrm{C}^{n-1}Z}{\mathrm{C}^{n-1}Z}\right)\right)\ge&2 \left(1 - \frac{3}{4} \lambda\right)^n - 10\cdot 2^{-n}(1 - \frac{\lambda}{2})^n\\
\ge&2 \left(1 - \frac{3}{4} \lambda\right)^n  - 0.5,
\end{aligned}
\end{equation}
where the last inequality holds when $n\ge5$.
For $\epsilon\ge0$, let $2 \left(1 - \frac{3}{4} \lambda_n\right)^n  - 0.5=1+\epsilon$, we have $\lambda_n=\frac{4}{3}\left(1-(0.75+\epsilon/2)^{\frac{1}{n}}\right)=\Omega(n^{-1})$. 
Combined with the analysis of upper bound in Section~\ref{section:HighDegreeEdgeFragile}, we have $\lambda^*_\epsilon=\Theta(n^{-1})$ for $\ket{\mathrm{C}^{n-1}Z}$.

\subsection{Local magic}

Note that
\begin{equation}
\begin{aligned}
\abs{\Tr\left(P_{(\mathbf{z},\mathbf{x})}\ketbra{+^n}{+^n}\right)}=
\begin{cases}
1, & \z=0^n,\\
0, & \z\neq0^n.\\
\end{cases}
\end{aligned}
\end{equation}
Denote $\rho_n=\Tr_{\overline{\{0,1,2\}}}(\ketbra{\mathrm{C}^{n-1}Z}{\mathrm{C}^{n-1}Z})=\left(1-\frac{1}{2^{n-3}}\right)\ketbra{+^3}{+^3}+\frac{1}{2^{n-3}}\ketbra{\mathrm{CC}Z}{\mathrm{CC}Z}$, then we have
\begin{equation}
\begin{aligned}
\abs{\Tr\left(P_{(\mathbf{z},\mathbf{x})}\rho_n\right)}=
\begin{cases}
1, & \x=0^3\text{ and }\z=0^3,\\
0, & \x=0^3\text{ and }\z\neq0^3,\\
\frac{1}{2^{n-3}}\left(1-\frac{1}{2}\right)+\left(1-\frac{1}{2^{n-3}}\right)=1-\frac{1}{2^{n-2}}, & \x\neq0^3\text{ and }\z=0^3,\\
\frac{1}{2^{n-3}}\frac{1}{2}=\frac{1}{2^{n-2}}, & \x\neq0^3\text{ and }\z\neq0^3\text{ and }\abs{\x\land \z}\text{ even,}\\
0, & \x\neq0^3\text{ and }\z\neq0^3\text{ and }\abs{\x\land \z}\text{ odd.}\\
\end{cases}
\end{aligned}
\end{equation}
Thus $\mc{D}(\rho_n)=1+\frac{7}{2^n}>1$. Note that $\mc{R}(\rho_n)\le1+1.55\frac{1}{2^{n-3}}$. 
Thus we have $\mc{R}(\rho_n)=1+\Theta(2^{-n})$.

\section{High-degree edges are fragile}\label{app:proof_of_HighDegreeEdgeFragile}

\begin{thm}\label{thm:HighDegreeEdgeFragile_app}
Let $\Psi$ be a $n$-qubit hypergraph state. 
Consider adding $K$ edges $e_1,\cdots,e_K$ to $\Psi$ to obtain $\Phi$, then we have
\begin{equation}\label{eq:HighDegreeEdgeFragile_app}
\begin{aligned}
\mc{R}\left(\mc{E}_{\lambda}^{\otimes n}\left(\Phi\right)\right)\le\mc{R}\left(\mc{E}_{\lambda}^{\otimes n}\left(\Psi\right)\right)+C_{\Psi}\sum_{\emptyset\neq J\subset[K]}(5^{\abs{J}}+1)\Big(1-\frac{\lambda}{2}\Big)^{\abs{\cup_{j\in J}e_j}},
\end{aligned}
\end{equation}
where $C_{\Psi}:=\underset{I\subset[n],\s\in\mbb{Z}_2^{\abs{I}}}{\max}~\mc{R}\big(\Psi^{(I,\mathbf{s})}\big)$ is a constant depending on $\Psi$.
\end{thm}

\begin{proof}
Note that
\begin{equation}
\begin{aligned}
\mc{E}_{\lambda}^{\otimes n}(\Psi)=&\left((1-\lambda)\mc{I}+\lambda\mc{G}\right)^{\otimes n}(\Psi)\\
=&\sum_{I\subset[n]}(1-\lambda)^{n-\abs{I}}\lambda^{\abs{I}}\Tr_I(\Psi)\otimes\frac{\mbb{I}_I}{2^{\abs{I}}}\\
=&\sum_{I\subset[n]}(1-\lambda)^{n-\abs{I}}\lambda^{\abs{I}}\frac{1}{2^{\abs{I}}}\sum_{\s\in\mbb{Z}_2^{\abs{I}}}\Psi^{(I,\s)}\otimes\frac{\mbb{I}_I}{2^{\abs{I}}}.
\end{aligned}
\end{equation}
Thus we have
\begin{equation}
\begin{aligned}
\mc{E}_{\lambda}^{\otimes n}(\Phi)-\mc{E}_{\lambda}^{\otimes n}(\Psi)=&\sum_{I\subset[n]}(1-\lambda)^{n-\abs{I}}\lambda^{\abs{I}}\left[\Tr_I(\Phi)-\Tr_I(\Psi)\right]\otimes\frac{\mbb{I}_I}{2^{\abs{I}}}\\
=&\sum_{I\subset[n]}(1-\lambda)^{n-\abs{I}}\lambda^{\abs{I}}\frac{1}{2^{\abs{I}}}\sum_{\s\in\mbb{Z}_2^{\abs{I}}}\left[\Phi^{(I,\s)}-\Psi^{(I,\s)}\right]\otimes\frac{\mbb{I}_I}{2^{\abs{I}}}.
\end{aligned}
\end{equation}
For $I\subset[n]$, $\s\in\mbb{Z}_2^{\abs{I}}$ and edge $e_j$, we denote $\s|_{e_j}\in\mbb{Z}_2^{\abs{I\cap e_j}}$ be the vector obtained by restricting $\s$ on $e_j$. Note that if $\s\in A_I$, where we define
\begin{equation}
A_I=\left\{\s\in\mbb{Z}_2^{\abs{I}}:\forall j\in[K],I\cap e_j\neq\emptyset\text{ and }\s|_{e_j}\neq\mathbf{1}\right\},
\end{equation}
then $\Phi^{(I,\s)}=\Psi^{(I,\s)}$. Thus denote $B_I=\mbb{Z}_2^{\abs{I}}-A_I$, we can write
\begin{equation}
\begin{aligned}
\mc{E}_{\lambda}^{\otimes n}(\Phi)-\mc{E}_{\lambda}^{\otimes n}(\Psi)=&\sum_{I\subset[n]}(1-\lambda)^{n-\abs{I}}\lambda^{\abs{I}}\frac{1}{2^{\abs{I}}}\sum_{\s\in B_I}\left[\Phi^{(I,\s)}-\Psi^{(I,\s)}\right]\otimes\frac{\mbb{I}_I}{2^{\abs{I}}}.
\end{aligned}
\end{equation}
By convexity of $\mc{R}(\cdot)$, we have
\begin{equation}
\begin{aligned}
\mc{R}\left(\mc{E}_{\lambda}^{\otimes n}(\Phi)\right)\le\mc{R}\left(\mc{E}_{\lambda}^{\otimes n}(\Psi)\right)+\sum_{I\subset[n]}(1-\lambda)^{n-\abs{I}}\lambda^{\abs{I}}\frac{1}{2^{\abs{I}}}\sum_{\s\in B_I}\left[\mc{R}\big(\Phi^{(I,\s)}\big)+\mc{R}\big(\Psi^{(I,\s)}\big)\right].
\end{aligned}
\end{equation}
We write
\begin{equation}
\begin{aligned}
B_I=&\left\{\s\in\mbb{Z}_2^{\abs{I}}:\exists j\in[K],I\cap e_j=\emptyset\text{ or }(I\cap e_j\neq\emptyset\text{ and }\s|_{e_j}=\mathbf{1})\right\}=\bigcup_{\emptyset\neq J\subset[K]}C_{I,J},
\end{aligned}
\end{equation}
where we define
\begin{equation}
\begin{aligned}
C_{I,J}=\Big\{\s\in\mbb{Z}_2^{\abs{I}}:\forall j\in J,I\cap e_j=\emptyset\text{ or }(I\cap e_j\neq\emptyset\text{ and }\s|_{e_j}=\mathbf{1});
\forall j\in[K]-J,I\cap e_j\neq\emptyset\text{ and }\s|_{e_j}\neq\mathbf{1}\Big\}.
\end{aligned}
\end{equation}
Now we have
\begin{equation}
\begin{aligned}
\mc{R}\left(\mc{E}_{\lambda}^{\otimes n}(\Phi)\right)\le\mc{R}\left(\mc{E}_{\lambda}^{\otimes n}(\Psi)\right)+\sum_{I\subset[n]}(1-\lambda)^{n-\abs{I}}\lambda^{\abs{I}}\frac{1}{2^{\abs{I}}}\sum_{\emptyset\neq J\subset[K]}\sum_{\s\in C_{I,J}}\left[\mc{R}\big(\Phi^{(I,\s)}\big)+\mc{R}\big(\Psi^{(I,\s)}\big)\right].
\end{aligned}
\end{equation}
For $\s\in C_{I,J}$, we can add at most $\abs{J}$ edges on $\Psi^{(I,\s)}$ to obtain $\Phi^{(I,\s)}$.
Using Lemma~\ref{lem:CCCZ_magic_UB} iteratively we obtain $\mc{R}\big(\Phi^{(I,\s)}\big)\le5^{\abs{J}}\mc{R}\big(\Psi^{(I,\s)}\big)$, thus
\begin{equation}
\begin{aligned}
\mc{R}\left(\mc{E}_{\lambda}^{\otimes n}(\Phi)\right)\le&\mc{R}\left(\mc{E}_{\lambda}^{\otimes n}(\Psi)\right)+\sum_{I\subset[n]}(1-\lambda)^{n-\abs{I}}\lambda^{\abs{I}}\frac{1}{2^{\abs{I}}}\sum_{\emptyset\neq J\subset[K]}\sum_{\s\in C_{I,J}}(5^{\abs{J}}+1)\mc{R}\big(\Psi^{(I,\s)}\big)\\
\le&\mc{R}\left(\mc{E}_{\lambda}^{\otimes n}(\Psi)\right)+C_{\Psi}\sum_{I\subset[n]}(1-\lambda)^{n-\abs{I}}\lambda^{\abs{I}}\frac{1}{2^{\abs{I}}}\sum_{\emptyset\neq J\subset[K]}\sum_{\s\in C_{I,J}}(5^{\abs{J}}+1)\\
=&\mc{R}\left(\mc{E}_{\lambda}^{\otimes n}(\Psi)\right)+C_{\Psi}\sum_{\emptyset\neq J\subset[K]}(5^{\abs{J}}+1)\sum_{I\subset[n]}(1-\lambda)^{n-\abs{I}}\lambda^{\abs{I}}\frac{1}{2^{\abs{I}}}\abs{C_{I,J}}\\
=&\mc{R}\left(\mc{E}_{\lambda}^{\otimes n}(\Psi)\right)+C_{\Psi}\sum_{\emptyset\neq J\subset[K]}(5^{\abs{J}}+1)\sum_{I_1\subset\cup_{j\in J}e_j}\sum_{I_2\subset[n]-\cup_{j\in J}e_j}(1-\lambda)^{n-\abs{I_1}-\abs{I_2}}\lambda^{\abs{I_1}+\abs{I_2}}\frac{1}{2^{\abs{I_1}+\abs{I_2}}}\abs{C_{I_1\cup I_2,J}},
\end{aligned}
\end{equation}
where we split $I$ into two disjoint sets $I_1$ and $I_2$ in the last line. Since $\s\in C_{I_1\cup I_2,J}$ can only take value $1$ in $I_1$, we know that $\abs{C_{I_1\cup I_2,J}}\le2^{\abs{I_2}}$, thus
\begin{equation}
\begin{aligned}
\mc{R}\left(\mc{E}_{\lambda}^{\otimes n}(\Phi)\right)\le&\mc{R}\left(\mc{E}_{\lambda}^{\otimes n}(\Psi)\right)+C_{\Psi}\sum_{\emptyset\neq J\subset[K]}(5^{\abs{J}}+1)\sum_{I_1\subset\cup_{j\in J}e_j}\sum_{I_2\subset[n]-\cup_{j\in J}e_j}(1-\lambda)^{n-\abs{I_1}-\abs{I_2}}\lambda^{\abs{I_1}+\abs{I_2}}\frac{1}{2^{\abs{I_1}}}.
\end{aligned}
\end{equation}
Denote $a=\abs{I_1}$ and $b=\abs{I_2}$, we have
\begin{equation}
\begin{aligned}
&\sum_{I_1\subset\cup_{j\in J}e_j}\sum_{I_2\subset[n]-\cup_{j\in J}e_j}(1-\lambda)^{n-\abs{I_1}-\abs{I_2}}\lambda^{\abs{I_1}+\abs{I_2}}\frac{1}{2^{\abs{I_1}}}\\
=&\sum_{a=0}^{\abs{\cup_{j\in J}e_j}}\binom{\abs{\cup_{j\in J}e_j}}{a}(1-\lambda)^{\abs{\cup_{j\in J}e_j}-a}\Big(\frac{\lambda}{2}\Big)^a\sum_{b=0}^{n-\abs{\cup_{j\in J}e_j}}\binom{n-\abs{\cup_{j\in J}e_j}}{b}(1-\lambda)^{n-\abs{\cup_{j\in J}e_j}-b}\lambda^b\\
=&\Big(1-\frac{\lambda}{2}\Big)^{\abs{\cup_{j\in J}e_j}}1^{n-\abs{\cup_{j\in J}e_j}}=\Big(1-\frac{\lambda}{2}\Big)^{\abs{\cup_{j\in J}e_j}},
\end{aligned}
\end{equation}
thus we obtain
\begin{equation}
\mc{R}\left(\mc{E}_{\lambda}^{\otimes n}\left(\Phi\right)\right)\le\mc{R}\left(\mc{E}_{\lambda}^{\otimes n}\left(\Psi\right)\right)+C_{\Psi}\sum_{\emptyset\neq J\subset[K]}(5^{\abs{J}}+1)\Big(1-\frac{\lambda}{2}\Big)^{\abs{\cup_{j\in J}e_j}}.
\end{equation}

\end{proof}

\section{Non-vanishing local magic implies non-vanishing threshold}\label{app:proof_LocalMagic_then_Threshold}

\begin{lem}\label{lem:Distance_LB_by_RoM}
Suppose $a$ and $b$ are two constant satisfying $b>a\ge1$,
then there exists a constant $C'>0$ such that for all $K$-qubit state $\rho$ satisfying $\mc{R}(\rho)\ge b$, and all $\sigma$ satisfying $\mc{R}(\sigma)\le a$, we have $\onenorm{\rho-\sigma}\ge C'$.
\end{lem}
\begin{proof}
If not, then there exist two sequences of states $\{\rho_m\}$ and $\{\sigma_m\}$ satisfying $\mc{R}(\rho_m)\ge b$, $\mc{R}(\sigma_m)\le a$ and $\onenorm{\rho_m-\sigma_m}\le\frac{1}{m}$. Denote the positive and negative part of $\rho_m-\sigma_m$ be $\tau_m^+$ and $\tau_m^-$. We have $\rho_m-\sigma_m=\tau_m^+-\tau_m^-$ and $\Tr(\tau_m^+)=\Tr(\tau_m^-)=\onenorm{\rho_m-\sigma_m}/2\le\frac{1}{2m}$. We can write $\rho_m=\sigma_m+\Tr(\tau_m^+)\frac{\tau_m^+}{\Tr(\tau_m^+)}-\Tr(\tau_m^+)\frac{\tau_m^-}{\Tr(\tau_m^-)}$. Note that the RoM of all $K$-qubit quantum states is upper bounded by $\sqrt{2^K(2^K+1)}$ \cite{Liu2022manybody}, thus by convexity of $\mc{R}(\cdot)$ we have
\begin{equation}
\begin{aligned}
\mc{R}(\rho_m)\le&\mc{R}(\sigma_m)+\Tr(\tau_m^+)\mc{R}\left(\frac{\tau_m^+}{\Tr(\tau_m^+)}\right)+\Tr(\tau_m^+)\mc{R}\left(\frac{\tau_m^-}{\Tr(\tau_m^-)}\right)\\
\le&a+\frac{1}{m}\sqrt{2^K(2^K+1)}<b,
\end{aligned}
\end{equation}
where the last inequality holds when $m$ is sufficiently large. This contradicts the previous assumption.
\end{proof}

\begin{prop}\label{prop:LocalMagic_then_Threshold_app}
Suppose $\{\rho_n\}$ is a family of states with non-vanishing local magic, then for any small enough $\epsilon\ge0$, it has a non-vanishing magic threshold $\lambda^*_\epsilon$.
\end{prop}
\begin{proof}
By definition of non-vanishing local magic, there exists a constant $K$ such that
\begin{equation}
\underset{n\rightarrow\infty}{\liminf}~\max_{J\subset[n],\abs{J}=K}\mc{R}\left(\Tr_{\overline{J}}(\rho_n)\right)>1.
\end{equation}
Thus there exists a constant $C>1$, and a sequence ${J_n}$ with $J_n\subset[n]$ and $\abs{J_n}=K$, such that $\mc{R}\left(\Tr_{\overline{J_n}}(\rho_n)\right)\ge C$ for all sufficiently large $n$.
Note that we have
\begin{equation}
\mc{R}\left(\mc{E}_\lambda^{\otimes n}\left(\rho_n\right)\right)\ge\mc{R}\left(\Tr_{\overline{J_n}}\left(\mc{E}_\lambda^{\otimes n}\left(\rho_n\right)\right)\right)=\mc{R}\left(\mc{E}_\lambda^{\otimes K}\left(\Tr_{\overline{J_n}}\left(\rho_n\right)\right)\right),
\end{equation}
it suffices to show that for all sufficiently large $n$, the state $\Tr_{\overline{J_n}}\left(\rho_n\right)$ has a constant (independent of $n$) threshold.

Take $b=C>1+\epsilon=a$ in Lemma~\ref{lem:Distance_LB_by_RoM}, we know that there exists a constant $C'>0$ such that for all sufficiently large $n$ and all $\sigma$ satisfying $\mc{R}(\sigma)\le1+\epsilon$, we have $\onenorm{\Tr_{\overline{J_n}}\left(\rho_n\right)-\sigma}\ge C'$.

But the distance when moving in the space of $K$-qubit state under noise $\mc{E}_\lambda^{\otimes K}$ has an upper bound with regard to $\lambda$. Specifically, for arbitrary $K$-qubit state $\tau$, we have
\begin{equation}
\begin{aligned}
\onenorm{\mc{E}_\lambda^{\otimes K}(\tau)-\tau}\le&\left(1-(1-\lambda)^K\right)\onenorm{\tau}+\sum_{I\subset[n],\abs{I}\ge1}(1-\lambda)^{K-\abs{I}}\lambda^{\abs{I}}\onenorm{\Tr_{I}(\tau)\otimes\frac{\mbb{I}_I}{2^{\abs{I}}}-\tau}\\
\le&3\left(1-(1-\lambda)^K\right).
\end{aligned}
\end{equation}

To ensure $\mc{E}_\lambda^{\otimes K}\left(\Tr_{\overline{J_n}}\left(\rho_n\right)\right)$ falls into the region with RoM $\le1+\epsilon$, $\lambda$ must be big enough to guarantee $3\left(1-(1-\lambda)^K\right)\ge C'$, that is, $\lambda\ge1-\left(1-C'/3\right)^{\frac{1}{K}}$.

Thus the magic threshold of $\{\rho_n\}$ satisfies $\lambda^*_\epsilon\ge1-\left(1-C'/3\right)^{\frac{1}{K}}$, which is a constant with regard to $n$.
\end{proof}

\section{Example: $3$-complete hypergraph}\label{app:3-complete}

\subsection{Upper bound}

\begin{lem}\label{lemma:3complete_partialtrace}
Denote $\Gamma_n$ be the hypergraph state corresponding to the $3$-complete hypergraph with $n$ vertices. For $K\le n$ a constant, we have
\begin{equation}
\Tr_{\overline{[K]}}(\Gamma_n)=\frac{1}{2^{n-K}}\sum_{j=0}^{n-K}\binom{n-K}{j}\Phi_{j\bmod 4},
\end{equation}
where we define the four $K$-qubit hypergraph states as
\begin{equation}\label{eq:3complete_fourHGstate}
\begin{aligned}
\ket{\Phi_0}&=\ket{\Gamma_K},\\
\ket{\Phi_1}&=\prod_{e\subset[K],\abs{e}=2}\mathrm{C}Z_e\ket{\Gamma_K},\\
\ket{\Phi_2}&=\prod_{v\in[K]}Z_v\ket{\Gamma_K},\\
\ket{\Phi_3}&=\prod_{e\subset[K],\abs{e}=2}\mathrm{C}Z_e\prod_{v\in[K]}Z_v\ket{\Gamma_K}.
\end{aligned}
\end{equation}

\end{lem}
\begin{proof}
By Lemma~\ref{lemma:HGstate_PartialTrace}, we have
\begin{equation}
\Tr_{\overline{[K]}}(\Gamma_n)=\frac{1}{2^{n-K}}\sum_{\mathbf{b}\in\mbb{Z}_2^{n-K}}\Gamma_n^{(\overline{[K]},\mathbf{b})}.
\end{equation}
If a qubit is labeled $0$ when tracing out, then all the edges around this qubit are removed, thus $\Gamma_n^{(\overline{[K]},\mathbf{b})}=\Gamma_{K+\abs{\mathbf{b}}}^{(\overline{[K]},1^{\abs{\mathbf{b}}})}$, which only depends on $\abs{\mathbf{b}}$. Denote $\Sigma_j=\Gamma_{K+j}^{(\overline{[K]},1^j)}$ for $j=0,1,2,\cdots$ for clarity. Now we prove $\Sigma_j=\Phi_{j\bmod 4}$.

We have $\Sigma_0=\Phi_0$. 

For $\Sigma_1$, the qubit $K+1$ is traced out, thus the edges $\{k_1,k_2,K+1\}$ becomes $\{k_1,k_2\}$, for all $\{k_1,k_2\}\subset[K]$, adding all the $2$-edges to the remaining $K$ qubits. Thus $\Sigma_1=\Phi_1$. 

For $\Sigma_2$, the edges $\{k_1,k_2,K+1\}$, $\{k_1,k_2,K+2\}$ becomes $\{k_1,k_2\}$, for all $\{k_1,k_2\}\subset[K]$, which cancels themselves. The edges $\{k,K+1,K+2\}$ becomes $\{k\}$, for all $k\in [K]$, adding all the $1$-edges to the remaining $K$ qubits. Thus $\Sigma_2=\Phi_2$.

For $\Sigma_3$, the edges $\{k_1,k_2,K+1\}$, $\{k_1,k_2,K+2\}$, $\{k_1,k_2,K+3\}$ becomes $\{k_1,k_2\}$, for all $\{k_1,k_2\}\subset[K]$, which adds all the $2$-edges to the remaining $K$ qubits. The edges $\{k,K+1,K+2\}$, $\{k,K+1,K+3\}$, $\{k,K+2,K+3\}$ becomes $\{k\}$, for all $k\in [K]$, adding all the $1$-edges. Thus $\Sigma_3=\Phi_3$.

More generally, $\Sigma_j$ has the $2$-edges $\iff$ $j$ odd. $\Sigma_j$ has the $1$-edges $\iff$ $\binom{j}{2}$ odd $\iff$ when $j\overset{\bmod 4}{=}2\text{ or }3$. Thus we have $\Sigma_j=\Phi_{j\bmod 4}$.

Now we can compute $\Tr_{\overline{[K]}}(\Gamma_n)$ as
\begin{equation}
\begin{aligned}
\Tr_{\overline{[K]}}(\Gamma_n)
=\frac{1}{2^{n-K}}\sum_{\mathbf{b}\in\mbb{Z}_2^{n-K}}\Gamma_{K+\abs{\mathbf{b}}}^{(\overline{[K]},1^{\abs{\mathbf{b}}})}=\frac{1}{2^{n-K}}\sum_{j=0}^{n-K}\binom{n-K}{j}\Phi_{j\bmod 4}.
\end{aligned}
\end{equation}
\end{proof}

\begin{lem}\label{lemma:3complete_1/4mixed_is_STAB}
For hypergraph states $\Phi_l$ defined in Eq.~\eqref{eq:3complete_fourHGstate}, we have $\frac{1}{4}\sum_{l=0}^3\Phi_l\in\mathrm{SATB}$.
\end{lem}
\begin{proof}
We first show that $\frac{1}{4}\sum_{l=0}^3\Phi_l$ is invariant under $\prod_{e\subset[K],\abs{e}=3}\mathrm{CC}Z_e$.
For $\mathbf{a},\mathbf{b}\in\mbb{Z}_2^K$, we have 
\begin{equation}
	\begin{aligned}
		2^{K}\sum_{l=0}^3\bra{\mathbf{a}}\Phi_l\ket{\mathbf{b}}=&(-1)^{\binom{\abs{\mathbf{a}}}{3}+\binom{\abs{\mathbf{b}}}{3}}+(-1)^{\binom{\abs{\mathbf{a}}}{3}+\binom{\abs{\mathbf{b}}}{3}+\binom{\abs{\mathbf{a}}}{2}+\binom{\abs{\mathbf{b}}}{2}}\\
		&+(-1)^{\binom{\abs{\mathbf{a}}}{3}+\binom{\abs{\mathbf{b}}}{3}+\binom{\abs{\mathbf{a}}}{1}+\binom{\abs{\mathbf{b}}}{1}}+(-1)^{\binom{\abs{\mathbf{a}}}{3}+\binom{\abs{\mathbf{b}}}{3}+\binom{\abs{\mathbf{a}}}{2}+\binom{\abs{\mathbf{b}}}{2}+\binom{\abs{\mathbf{a}}}{1}+\binom{\abs{\mathbf{b}}}{1}}\\
		=&(-1)^{\binom{\abs{\mathbf{a}}}{3}+\binom{\abs{\mathbf{b}}}{3}}\left(1+(-1)^{\binom{\abs{\mathbf{a}}}{2}+\binom{\abs{\mathbf{b}}}{2}}\right)\left(1+(-1)^{\binom{\abs{\mathbf{a}}}{1}+\binom{\abs{\mathbf{b}}}{1}}\right).
	\end{aligned}
\end{equation}
Note that: $\binom{\abs{\mathbf{a}}}{3}+\binom{\abs{\mathbf{b}}}{3}$ odd $\Rightarrow$ one of $\abs{\mathbf{a}}$ and $\abs{\mathbf{b}}$ $\overset{\bmod4}{=}3$ $\Rightarrow$ one of $\binom{\abs{\mathbf{a}}}{2}+\binom{\abs{\mathbf{b}}}{2}$ and $\binom{\abs{\mathbf{a}}}{1}+\binom{\abs{\mathbf{b}}}{1}$ odd $\Rightarrow$ $\sum_{l=0}^3\bra{\mathbf{a}}\Phi_l\ket{\mathbf{b}}=0$. 
Thus
\begin{equation}
	\begin{aligned}
		\frac{1}{4}\sum_{l=0}^3\Phi_l=\sum_{\binom{\abs{\mathbf{a}}}{3}+\binom{\abs{\mathbf{b}}}{3}\text{ even}}\bra{\mathbf{a}}\left(\frac{1}{4}\sum_{l=0}^3\Phi_l\right)\ket{\mathbf{b}}\ketbra{\mathbf{a}}{\mathbf{b}},
	\end{aligned}
\end{equation}
which is invariant under $\prod_{e\subset[K],\abs{e}=3}\mathrm{CC}Z_e$, because $\prod_{e\subset[K],\abs{e}=3}\mathrm{CC}Z_e\ket{\mathbf{a}}=(-1)^{\binom{\abs{\mathbf{a}}}{3}}\ket{\mathbf{a}}$.
Denote $\ket{\Phi_l'}=\prod_{e\subset[K],\abs{e}=3}\mathrm{CC}Z_e\ket{\Phi_l}$, $l=0,1,2,3$, which are four pure stabilizer states. We can see that
\begin{equation}
	\begin{aligned}
		\frac{1}{4}\sum_{l=0}^3\Phi_l=\Bigg(\prod_{e\subset[K],\abs{e}=3}\mathrm{CC}Z_e\Bigg)\left(\frac{1}{4}\sum_{l=0}^3\Phi_l\right)\Bigg(\prod_{e\subset[K],\abs{e}=3}\mathrm{CC}Z_e\Bigg)=\frac{1}{4}\sum_{l=0}^3\Phi_l'
	\end{aligned}
\end{equation}
is a convex combination of these four pure stabilizer states, thus a mixed stabilizer state.
\end{proof}

\begin{prop}\label{prop:3complete_localmagic_UB}
Denote $\Gamma_n$ be the hypergraph state corresponding to the $3$-complete hypergraph with $n$ vertices. For $K$ a constant, we have
\begin{equation}
\mc{R}\left(\Tr_{\overline{[K]}}(\Gamma_n)\right)\le1+2^{1+\frac{3}{2}K-\frac{n}{2}}.
\end{equation}
\end{prop}
\begin{proof}
By Lemma~\ref{lemma:3complete_partialtrace}, we have
\begin{equation}
\Tr_{\overline{[K]}}(\Gamma_n)=\sum_{l=0}^3\left[\frac{1}{2^{n-K}}\sum_{j=0}^{n-K}\binom{n-K}{j}\delta_{j=l \bmod  4}\right]\Phi_{l}.
\end{equation}
Denote $m=n-K$, note that
\begin{equation}
\begin{aligned}
\frac{1}{2^m}\sum_{j=0}^m\binom{m}{j}\delta_{j=0 \bmod  4}&=\frac{1}{4}\left[\left(\frac{1+i}{2}\right)^m+\left(\frac{1-i}{2}\right)^m+1\right],\\
\frac{1}{2^m}\sum_{j=0}^m\binom{m}{j}\delta_{j=1 \bmod  4}&=\frac{1}{4}\left[-i \left(\frac{1+i}{2}\right)^m+i \left(\frac{1-i}{2}\right)^m+1\right],\\
\frac{1}{2^m}\sum_{j=0}^m\binom{m}{j}\delta_{j=2 \bmod  4}&=\frac{1}{4}\left[-\left(\frac{1+i}{2}\right)^m-\left(\frac{1-i}{2}\right)^m+1\right],\\
\frac{1}{2^m}\sum_{j=0}^m\binom{m}{j}\delta_{j=3 \bmod  4}&=\frac{1}{4}\left[i\left(\frac{1+i}{2}\right)^m-i\left(\frac{1-i}{2}\right)^m+1\right],
\end{aligned}
\end{equation}
where $i:=\sqrt{-1}$. Since they all tend to $1/4$ as $m\rightarrow\infty$, we have $\Tr_{\overline{[K]}}(\Gamma_n)\rightarrow\frac{1}{4}\sum_{l=0}^3\Phi_l \text{ when } n\rightarrow\infty$.
We can write
\begin{equation}
\begin{aligned}
\Tr_{\overline{[K]}}(\Gamma_n)=\frac{1}{4}\sum_{l=0}^3\Phi_l+\sum_{l=0}^3\left[\frac{1}{2^{m}}\sum_{j=0}^{m}\binom{m}{j}\delta_{j=l \bmod  4}-\frac{1}{4}\right]\Phi_{l},
\end{aligned}
\end{equation}
also note that $\mc{R}(\Phi_{l})=\mc{R}(\Gamma_K)$ for all $l$, thus
\begin{equation}
	\begin{aligned}
		\mc{R}\left(\Tr_{\overline{[K]}}(\Gamma_n)\right)&\le1+\mc{R}(\Gamma_K)\sum_{l=0}^3\abs{\frac{1}{2^m}\sum_{j=0}^m\binom{m}{j}\delta_{j=l \bmod  4}-\frac{1}{4}}\le1+\mc{R}(\Gamma_K)2^{\frac{1}{2}-\frac{m}{2}}\\
		&\le1+\sqrt{2^K(2^K+1)}2^{\frac{1}{2}-\frac{n-K}{2}}\le1+2^{1+\frac{3}{2}K-\frac{n}{2}}.
	\end{aligned}
\end{equation}
\end{proof}

Now we can provide an upper bound for the noisy magic of $3$-complete hypergraph state as
\begin{equation}
\begin{aligned}
\mc{R}\left(\mc{E}_{\lambda}^{\otimes n}(\Gamma_n)\right)\le&\sum_{J\subset[n]}(1-\lambda)^{\abs{J}}\lambda^{n-\abs{J}}\mc{R}\left(\Tr_{\overline{J}}(\Gamma_n)\right)\\
=&\sum_{k=0}^n(1-\lambda)^{k}\lambda^{n-k}\mc{R}\left(\Tr_{\overline{[k]}}(\Gamma_n)\right)\\
\le&\sum_{k=0}^n\binom{n}{k}(1-\lambda)^{k}\lambda^{n-k}\left(1+2^{1+\frac{3}{2}k-\frac{n}{2}}\right)\\
=&1+2\sum_{k=0}^n\binom{n}{k}(2-2\lambda)^{k}(2^{-\frac{1}{2}}\lambda)^{n-k}\\
=&1+2\left(2-(2-2^{-\frac{1}{2}})\lambda\right)^n.\\
\end{aligned}
\end{equation}
When $\lambda\ge(2-2^{-\frac{1}{2}})^{-1}\lesssim0.78$, the upper bound converge to $1$ exponentially with regrad to $n$. Thus for arbitrary $\epsilon>0$, the magic threshold $\lambda^*_\epsilon\le0.78$ when $n$ is sufficiently large.

\subsection{Lower bound}

Denote $\Gamma_n$ be the hypergraph states corresponding to the $3$-complete hypergraph with $n$ vertices. We consider the case that $n\in\{3,5,7,\cdots\}$ is an odd number. The case that $n$ is an even number is similar. 
\begin{lem}
For $n\in\{3,5,7,\cdots\}$, we have
\begin{equation}
\begin{aligned}
\abs{\Tr\left(P_{(\mathbf{z},\mathbf{x})}\Gamma_n\right)}=
\begin{cases}
1,&\x=\z=0^n,\\ 
\frac{1}{2}, &\abs{\x}=2,4,\cdots,n-1\text{ and }(\z=0^n\text{ or }\z=\x\text{ or }\z=1^n\text{ or }\z=\x+1^n),\\
2^{-\frac{n-1}{2}},&\abs{\x}=1,3,\cdots,n \text{ and }\x\cdot\z=0\\
0,&\text{otherwise}.
\end{cases}
\end{aligned}
\end{equation}
\end{lem}

\begin{proof}
We have $\ket{\Gamma_n}=\frac{1}{\sqrt{2^n}}\sum_{\mathbf{s}\in\mbb{Z}_2^n}(-1)^{f(\mathbf{s})}\ket{\mathbf{s}}$, where $f(\mathbf{s})=\sum_{e\subset[n],\abs{e}=3}\prod_{i\in e}s_i=\binom{\abs{\mathbf{s}}}{3}$.
By Eq.~\eqref{eq:Pauli_spectrum_abs}, we need to calculate
\begin{equation}
\frac{1}{2^n}\sum_{\s\in\mbb{Z}_2^n}(-1)^{f(\s)+f(\s+\x)+\z\cdot\s}=\frac{1}{2^n}\sum_{\s\in\mbb{Z}_2^n}(-1)^{\binom{\abs{\s}}{3}+\binom{\abs{\s+\x}}{3}+\z\cdot\s}
\end{equation}
for all $\x,\z$.

(\textbf{Case 1.} $\abs{\x}=0$.) We have $\frac{1}{2^n}\sum_{\s\in\mbb{Z}_2^n}(-1)^{f(\s)+f(\s+\x)+\z\cdot\s}=\frac{1}{2^n}\sum_{\s\in\mbb{Z}_2^n}(-1)^{\z\cdot\s}$, which is $1$ for $\abs{\z}=0$, and $0$ for $\abs{\z}>0$.

(\textbf{Case 2.} $\abs{\x}=1,3,\cdots,n$.) For integer $m\ge0$, we have
\begin{equation}\label{eq:choose3_evenodd}
\binom{m}{3}\overset{\bmod2}{=}
\begin{cases}
0,&\text{when }m\overset{\bmod4}{=}0,1,2,\\
1,&\text{when }m\overset{\bmod4}{=}3.
\end{cases}
\end{equation}
Note that
\begin{equation}
\abs{\s+\x}\overset{\bmod4}{=}
\begin{cases}
\abs{\s}+1,&\text{when }\abs{\x\land\s}-\frac{\abs{\x}-1}{2}\text{ even},\\
\abs{\s}-1,&\text{when }\abs{\x\land\s}-\frac{\abs{\x}-1}{2}\text{ odd},
\end{cases}
\end{equation}
thus by Eq.~\eqref{eq:choose3_evenodd},
\begin{equation}
\binom{\abs{\s}}{3}+\binom{\abs{\s+\x}}{3}\overset{\bmod2}{=}
\begin{cases}
0,&\text{when }\abs{\s}\overset{\bmod4}{=}0,\text{ and }\abs{\x\land\s}-\frac{\abs{\x}-1}{2}\text{ even},\\
0,&\text{when }\abs{\s}\overset{\bmod4}{=}1,\\
0,&\text{when }\abs{\s}\overset{\bmod4}{=}2,\text{ and }\abs{\x\land\s}-\frac{\abs{\x}-1}{2}\text{ odd},\\
1,&\text{ otherwise}.
\end{cases}
\end{equation}
We also have $\z\cdot\s\overset{\bmod2}{=}\abs{\z\land\s}$.
Denote $W_I,W_X,W_Y,W_Z$ be the number of Pauli $\mbb{I},X,Y,Z$ in the Pauli word $P_{(\z,\x)}$, satisfying $W_I+W_X+W_Y+W_Z=n$. Denote $a,b,c,d$ be the number of $1$'s in $\s$ such that the corresponding position of $P_{(\z,\x)}$ is $Y,X,Z,I$, satisfying $a+b+c+d=\abs{\s}$. For such $\s$, the indicator of $\binom{\abs{\s}}{3}+\binom{\abs{\s+\x}}{3}\overset{\bmod2}{=}0$ can be written as
\begin{equation}
I(a,b,c,d):=\delta_{a+b+c+d\overset{\bmod4}{=}0}\delta_{a+b-\frac{\abs{\x}-1}{2}\text{ even}}+\delta_{a+b+c+d\overset{\bmod4}{=}1}+\delta_{a+b+c+d\overset{\bmod4}{=}2}\delta_{a+b-\frac{\abs{\x}-1}{2}\text{ odd}}.
\end{equation}
We count the number of $\s$ such that  $\binom{\abs{\s}}{3}+\binom{\abs{\s+\x}}{3}+\z\cdot\s\overset{\bmod2}{=}0$ as
\begin{equation}\label{eq:3complete_0count_case2}
\begin{aligned}
\sum_{a=0}^{W_Y}\sum_{b=0}^{W_X}\sum_{c=0}^{W_Z}\sum_{d=0}^{W_I}\binom{W_Y}{a}\binom{W_X}{b}\binom{W_Z}{c}\binom{W_I}{d}\Big[\delta_{a+c\text{ even}}I(a,b,c,d)+\delta_{a+c\text{ odd}}\left(1-I(a,b,c,d)\right)\Big],
\end{aligned}
\end{equation}
thus
\begin{equation}\label{eq:3complete_0count_case2_2}
\begin{aligned}
&\frac{1}{2^n}\sum_{\s\in\mbb{Z}_2^n}(-1)^{\binom{\abs{\s}}{3}+\binom{\abs{\s+\x}}{3}+\z\cdot\s}=2^{1-n}\eqref{eq:3complete_0count_case2}-1\\
=&i^{1-W_X}2^{-\frac{n}{2}}\left[\frac{1-(-1)^{W_X}}{2}e^{-i\frac{\pi}{4}(n-2W_I)}+\frac{1+(-1)^{W_Y}}{2}e^{i\frac{\pi}{4}(n-2W_I)}\right].
\end{aligned}
\end{equation}
In this case, we assume $\abs{\x}=W_X+W_Y$ odd, and we have $W_Y\text{ even}\iff\x\cdot\z=0$. Also, $n-2W_I$ is odd. Taking the absolute value we obtain
\begin{equation}
\begin{aligned}
\abs{\eqref{eq:3complete_0count_case2_2}}=
\begin{cases}
2^{\frac{1}{2}-\frac{n}{2}},&\text{ when }\x\cdot\z=0,\\
0,&\text{ when }\x\cdot\z=1.
\end{cases}
\end{aligned}
\end{equation}

(\textbf{Case 3.} $\abs{\x}=2,4,\cdots,n-1$.)
Note that
\begin{equation}
\abs{\s+\x}\overset{\bmod4}{=}
\begin{cases}
\abs{\s},&\text{when }\abs{\x\land\s}-\frac{\abs{\x}}{2}\text{ even},\\
\abs{\s}+2,&\text{when }\abs{\x\land\s}-\frac{\abs{\x}}{2}\text{ odd},
\end{cases}
\end{equation}
thus by Eq.~\eqref{eq:choose3_evenodd},
\begin{equation}
\binom{\abs{\s}}{3}+\binom{\abs{\s+\x}}{3}\overset{\bmod2}{=}
\begin{cases}
0,&\text{when }\abs{\x\land\s}-\frac{\abs{\x}}{2}\text{ even},\\
0,&\text{when }\abs{\x\land\s}-\frac{\abs{\x}}{2}\text{ odd},\text{ and }\abs{\s}\text{ even},\\
1,&\text{ otherwise}.
\end{cases}
\end{equation}
Again, denote $W_I,W_X,W_Y,W_Z$ be the number of Pauli $\mbb{I},X,Y,Z$ in the Pauli word $P_{(\z,\x)}$. Denote $a,b,c,d$ be the number of $1$'s in $\s$ such that the corresponding position of $P_{(\z,\x)}$ is $Y,X,Z,I$. For such $\s$, the indicator of $\binom{\abs{\s}}{3}+\binom{\abs{\s+\x}}{3}\overset{\bmod2}{=}0$ can be written as
\begin{equation}
J(a,b,c,d):=\delta_{a+b-\frac{\abs{\x}}{2}\text{ even}}+\delta_{a+b-\frac{\abs{\x}}{2}\text{ odd}}\delta_{a+b+c+d\text{ even}}.
\end{equation}
 We count the number of $\s$ such that  $\binom{\abs{\s}}{3}+\binom{\abs{\s+\x}}{3}+\z\cdot\s\overset{\bmod2}{=}0$ as
\begin{equation}\label{eq:3complete_0count_case3}
\begin{aligned}
&\sum_{a=0}^{W_Y}\sum_{b=0}^{W_X}\sum_{c=0}^{W_Z}\sum_{d=0}^{W_I}\binom{W_Y}{a}\binom{W_X}{b}\binom{W_Z}{c}\binom{W_I}{d}\Big[\delta_{a+c\text{ even}}J(a,b,c,d)+\delta_{a+c\text{ odd}}\left(1-J(a,b,c,d)\right)\Big]\\
&=
\begin{cases}
3\cdot2^{-2+n},&\text{when }W_Y=W_Z=0\text{ or }W_I=W_X=0,\\
(2+i^{-\abs{\x}})2^{-2+n},&\text{when }W_X=W_Z=0,\\
(2-i^{-\abs{\x}})2^{-2+n},&\text{when }W_I=W_Y=0,\\
2^{-1+n},&\text{otherwise}.
\end{cases}
\end{aligned}
\end{equation}
Thus
\begin{equation}\label{eq:3complete_0count_case3_2}
\begin{aligned}
&\Big|\frac{1}{2^n}\sum_{\s\in\mbb{Z}_2^n}(-1)^{\binom{\abs{\s}}{3}+\binom{\abs{\s+\x}}{3}+\z\cdot\s}\Big|=\abs{2^{1-n}\eqref{eq:3complete_0count_case3}-1}\\
=&
\begin{cases}
\frac{1}{2},&\text{when }\z=0^n\text{ or }\z=\x\text{ or }\z=1^n\text{ or }\z=\x+1^n,\\
0,&\text{otherwise}.
\end{cases}
\end{aligned}
\end{equation}

\end{proof}

Now for $n$ odd, we can use Eq.~\eqref{eq:LB} to provide a lower bound as
\begin{equation}
\begin{aligned}
\mc{D}\left(\mc{E}_\lambda^{\otimes n}\left(\Gamma_n\right)\right)=&\frac{1}{2^n}\sum_{\x,\z}\abs{\Tr\left(P_{(\mathbf{z},\mathbf{x})}\mc{E}_\lambda^{\otimes n}\left(\Gamma_n\right)\right)}\\
=&\frac{1}{2^n}\sum_{\x,\z}(1-\lambda)^{\abs{\x\vee \z}}\abs{\Tr\left(P_{(\mathbf{z},\mathbf{x})}\Gamma_n\right)}\\
=&\frac{1}{2^n}\Bigg[1+\sum_{\substack{\abs{\x}=2,4,\cdots,n-1\\\z=0^n \text{ or }\z=\x}}(1-\lambda)^{\abs{\x}}\frac{1}{2}+\sum_{\substack{\abs{\x}=2,4,\cdots,n-1\\\z=1^n\text{ or }\z=\x+1^n}}(1-\lambda)^{n}\frac{1}{2}+\sum_{\substack{\abs{\x}\text{ odd}\\\x\cdot\z=0}}(1-\lambda)^{\abs{\x\vee \z}}2^{-\frac{n-1}{2}}\Bigg]\\
=&\frac{1}{2^n}\Bigg[1+\sum_{k=2}^{n-1}\binom{n}{k}\frac{1+(-1)^k}{2}(1-\lambda)^{k}+\sum_{k=2}^{n-1}\binom{n}{k}\frac{1+(-1)^k}{2}(1-\lambda)^{n}\\
&+2^{-\frac{n-1}{2}}\sum_{k=1}^n\sum_{m=0}^k\sum_{l=m}^{n-k+m}\binom{n}{k}\binom{k}{m}\binom{n-k}{l-m}\frac{1-(-1)^k}{2}\frac{1+(-1)^m}{2}(1-\lambda)^{k+l-m}\Bigg]\\
=&\left(- 2^{-n} + 2^{-1}\right) (1 - \lambda)^n +2^{-n-1} (2 - \lambda)^n +\left(2^{-n-1}- 2^{-\frac{3}{2} - \frac{3}{2}n}\right)\lambda^n + 2^{-\frac{3}{2} - \frac{3}{2}n}\left(4-3\lambda\right)^n\\
\ge&2^{-\frac{3}{2} - \frac{3}{2}n}\left(4-3\lambda\right)^n-2^{-n}- 2^{-\frac{3}{2} - \frac{3}{2}n}.
\end{aligned}
\end{equation}
In the fourth and fifth lines, 
$k$ represents $\abs{\x}$, $m$ represents $\abs{\x\land \z}$, $l$ represents $\abs{\z}$.

For $\epsilon\ge0$, let $2^{-\frac{3}{2} - \frac{3}{2}n}\left(4-3\lambda_n\right)^n-2^{-n}- 2^{-\frac{3}{2} - \frac{3}{2}n}=1+\epsilon$, we have $\lambda_n=\frac{4-\left(2^{n/2+3/2}+2^{3n/2+3/2}+1+\epsilon\right)^{1/n}}{3}\rightarrow\frac{4-2^{3/2}}{3}\gtrsim 0.39$.

\subsection{Local magic}
Fix arbitrary constant $K$ in Proposition~\ref{prop:3complete_localmagic_UB}, we obtain $\mc{R}\left(\Tr_{\overline{[K]}}(\Gamma_n)\right)\le1+\mc{O}\left(2^{-\frac{n}{2}}\right)$.

\section{Qudit cases}

\subsection{Discrete Wigner function}\label{app:DiscreteWignerfunction}

For an odd prime dimension $d$, define operators $X,Z\in\mc{L}(\mbb{C}^d)$ by
\begin{equation}
\begin{aligned}
& X\ket{j}=\ket{j+1~\operatorname{mod}~d}, \\
& Z\ket{j}=\omega^j\ket{j},
\end{aligned}
\end{equation}
where $\omega=e^{2\pi i/d}$, $i=\sqrt{-1}$.
We have $ZX=\omega XZ$. For $(z,x)\in\mathbb{Z}_d \times \mathbb{Z}_d$, define the Heisenberg–Weyl (generalized Pauli) operator $T_{(z,x)}\in\mc{L}(\mbb{C}^d)$ by
\begin{equation}
T_{\left(z, x\right)}=\tau^{-z x} Z^{z} X^{x},
\end{equation}
where $\tau=e^{(d+1)\pi i/d}$. For $(\mathbf{z},\mathbf{x})\in\mbb{Z}_d^n\times\mbb{Z}_d^n$, where $\mathbf{z}=(z_1,\cdots,z_n)$ and $\mathbf{x}=(x_1,\cdots,x_n)$, define
\begin{equation}
T_{(\mathbf{z},\mathbf{x})}=T_{(z_1,x_1)}\otimes T_{(z_2,x_2)}\otimes\cdots\otimes T_{(z_n,x_n)}.
\end{equation}
Define $A_{(0,0)}\in\mc{L}(\mbb{C}^d)$ by
\begin{equation}
A_{(0,0)}=\frac{1}{d} \sum_{z,x\in\mbb{Z}_d} T_{(z,x)}=\sum_{j\in\mbb{Z}_d}\ketbra{-j}{j}.
\end{equation}
For $(z,x)\in\mbb{Z}_d\times\mbb{Z}_d$, let
\begin{equation}
A_{(z,x)}=T_{(z,x)} A_{(0,0)} T_{(z,x)}^{\dagger},
\end{equation}
whose eigenvalues are $\pm1$ and has trace $1$ \cite{Feng2024discrete}. For $(\mathbf{z},\mathbf{x})\in\mbb{Z}_d^n\times\mbb{Z}_d^n$, where $\mathbf{z}=(z_1,\cdots,z_n)$ and $\mathbf{x}=(x_1,\cdots,x_n)$, define
\begin{equation}
A_{(\mathbf{z},\mathbf{x})}=A_{(z_1,x_1)}\otimes A_{(z_2,x_2)}\otimes\cdots\otimes A_{(z_n,x_n)}.
\end{equation}
For an $n$-qudit state $\rho\in D\left((\mbb{C}^d)^{\otimes n}\right)$, and $\mathbf{u}\in\mbb{Z}_d^n\times\mbb{Z}_d^n$, define the corresponding discrete Wigner quasiprobability representation
\begin{equation}
W_\rho(\mathbf{u})=\frac{1}{d^n} \operatorname{Tr}( A_{\mathbf{u}} \rho).
\end{equation}
The Wigner negativity \cite{Veitch2014resource} of $\rho$ is defined by
\begin{equation}
\begin{aligned}
\operatorname{sn}(\rho):=\sum_{\mathbf{u}: W_\rho(\mathbf{u})<0}\left|W_\rho(\mathbf{u})\right| =\frac{1}{2}\left(\sum_{\mathbf{u}}\abs{W_\rho(\mathbf{u})}-1\right),
\end{aligned}
\end{equation}
which can provide a lower bound for the RoM of $\rho$ \cite{Liu2022manybody}:
\begin{equation}
2\operatorname{sn}(\rho)+1\le\mc{R}(\rho).
\end{equation}

\subsection{The maximum noise threshold for Wigner negativity}\label{app:maximumnoisethresholdforWignernegativity}

\begin{thm}
For an odd prime dimension $d$, and $\lambda_*=\frac{d}{d+1}$, we have:
\begin{enumerate}
    \item $\operatorname{sn}\left(\mc{E}_{\lambda}^{\otimes n}(\rho)\right)=0$, for arbitrary $n$-qudit state $\rho$ and $\lambda\ge\lambda_*$;
    \item There exists a $1$-qudit state $\rho$, such that $\operatorname{sn}\left(\mc{E}_\lambda(\rho)\right)>0$ for all $\lambda<\lambda_*$.
\end{enumerate}
\end{thm}
\begin{proof}
(Part 1)

For $\mathbf{u}\in\mbb{Z}_d^n\times\mbb{Z}_d^n$, we have
\begin{equation}
\begin{aligned}
W_{\mc{E}_{\lambda}^{\otimes n}(\rho)}(\mathbf{u})=\frac{1}{d^n}\Tr\left(A_{\mathbf{u}} \mc{E}_{\lambda}^{\otimes n}(\rho)\right)=\frac{1}{d^n}\Tr\left(\mc{E}_{\lambda}^{\otimes n}\left(A_{\mathbf{u}}\right)\rho\right)=\frac{1}{d^n}\Tr\left(T_{\mathbf{u}}\mc{E}_{\lambda}^{\otimes n}\left(A_{\mathbf{0}}\right)T_{\mathbf{u}}^{\dagger}\rho\right),
\end{aligned}
\end{equation}
where $\mathbf{0}$ denotes the all-zero element in $\mbb{Z}_d^n\times\mbb{Z}_d^n$, and the third equality holds because $\mc{E}_\lambda^{\otimes n}$ commutes with the tensor product of single qubit unitaries.
Note that $\mc{E}_{\lambda}^{\otimes n}\left(A_{\mathbf{0}}\right)=\mc{E}_\lambda\left(A_{(0,0)}\right)^{\otimes n}$, and
\begin{equation}
\begin{aligned}
\mc{E}_\lambda\left(A_{(0,0)}\right)=\lambda\frac{\mbb{I}_d}{d}+(1-\lambda)A_{(0,0)},
\end{aligned}
\end{equation}
whose eigenvalues are $1-\frac{d-1}{d}\lambda$ and $\frac{d+1}{d}\lambda-1$. 
When $\lambda\ge\lambda_*$, we find that $\mc{E}_\lambda\left(A_{(0,0)}\right)$ becomes a positive matrix, thus $W_{\mc{E}_{\lambda}^{\otimes n}(\rho)}(\mathbf{u})\ge0$ for all $\mathbf{u}$.

(Part 2)

Consider the $1$-qudit state $\rho=\ketbra{\psi}{\psi}$, where $\ket{\psi}=\frac{1}{\sqrt{2}}(\ket{1}-\ket{d-1})$. We have
\begin{equation}
\begin{aligned}
W_{\mc{E}_{\lambda}^{\otimes n}(\rho)}\big((0,0)\big)=\frac{1}{d}\Tr\left(\mc{E}_\lambda\left(A_{(0,0)}\right)\rho\right)=\frac{\lambda}{d^2}-\frac{1-\lambda}{d}=\frac{(d+1)\lambda-d}{d^2},
\end{aligned}
\end{equation}
which remains negative when $\lambda<\frac{d}{d+1}=\lambda_*$.
\end{proof}

\subsection{Qudit hypergraph states}\label{app:Qudithypergraphstates}

For $d$ is a prime, define the $n$-qudit multi-controlled-$Z$ gate $\mathrm{C}^{n-1}Z$ as
\begin{equation}
\mathrm{C}^{n-1}Z\ket{\mathbf{b}}=\omega^{\prod_{j=1}^nb_j}\ket{\mathbf{b}},
\end{equation}
where $\omega=e^{\frac{2\pi i}{d}}$ and $\mathbf{b}=(b_1,\cdots,b_n)\in\mbb{Z}_d^n$.
For a hypergraph $G=\left\{[n],E\right\}$, where $E\subset[n]$, and a set $\{\alpha_e\}_{e\in E}$, where $\alpha_e\in\mbb{Z}_d$ is the multiplicity of hyperedge $e$, we define the corresponding qudit hypergraph state $\ket{\Psi}$ as
\begin{equation}
\ket{\Psi}=\prod_{e\in E}\left(\mathrm{C}^{\abs{e}-1}Z_e\right)^{\alpha_e}\ket{+^n},
\end{equation}
where $\mathrm{C}^{\abs{e}-1}Z_e$ is the qudit multi-controlled-$Z$ gate applied on qudits in the set $e$, and $\ket{+}=\frac{1}{\sqrt{d}}\sum_{a\in\mbb{Z}_d}\ket{a}$.

\subsection{Wigner negativity of qudit hypergraph states}

Let $\ket{\Psi}=\frac{1}{\sqrt{d^n}}\sum_{\mathbf{s}\in\mbb{Z}_d^n}\omega^{f(\mathbf{s})}\ket{\mathbf{s}}$ be an $n$-qudit hypergraph state, where $\omega=e^{\frac{2\pi i}{d}}$, $i=\sqrt{-1}$, and $f:\mbb{Z}_d^n\rightarrow\mbb{Z}_d$ is the characteristic function of $\ket{\Psi}$.
For $\mathbf{z}=(z_1,\cdots,z_n)$ and $\mathbf{x}=(x_1,\cdots,x_n)$, we have
\begin{equation}
\begin{aligned}
W_{\Psi}(\mathbf{z},\mathbf{x})=&\frac{1}{d^n}\Tr\left(A_{(\mathbf{z},\mathbf{x})}\ketbra{\Psi}{\Psi}\right)\\
=&\frac{1}{d^n}\bra{\Psi}T_{(\mathbf{z},\mathbf{x})}A_{(0,0)}^{\otimes n}T_{(\mathbf{z},\mathbf{x})}^\dagger\ket{\Psi}.
\end{aligned}
\end{equation}
We can calculate $T^\dagger_{(\z,\x)}\ket{\Psi}$ as
\begin{equation}
\begin{aligned}
T^\dagger_{(\mathbf{z},\mathbf{x})}\ket{\Psi}=&e^{i\theta}\left(\otimes_{j=1}^nZ^{-z_j}\right)\left(\otimes_{j=1}^nX^{-x_j}\right)\ket{\Psi}\\
=&e^{i\theta}\left(\otimes_{j=1}^nZ^{-z_j}\right)\left(\otimes_{j=1}^nX^{-x_j}\right)\frac{1}{\sqrt{d^n}}\sum_{\mathbf{s}\in\mbb{Z}_d^n}\omega^{f(\mathbf{s})}\ket{\mathbf{s}}\\
=&e^{i\theta}\left(\otimes_{j=1}^nZ^{-z_j}\right)\frac{1}{\sqrt{d^n}}\sum_{\mathbf{s}\in\mbb{Z}_d^n}\omega^{f(\mathbf{s})}\ket{\mathbf{s}-\mathbf{x}}\\
=&e^{i\theta}\left(\otimes_{j=1}^nZ^{-z_j}\right)\frac{1}{\sqrt{d^n}}\sum_{\mathbf{s}\in\mbb{Z}_d^n}\omega^{f(\mathbf{s}+\mathbf{x})}\ket{\mathbf{s}}\\
=&e^{i\theta}\frac{1}{\sqrt{d^n}}\sum_{\mathbf{s}\in\mbb{Z}_d^n}\omega^{f(\mathbf{s}+\mathbf{x})-\mathbf{z}\cdot\mathbf{s}}\ket{\mathbf{s}},
\end{aligned}
\end{equation}
where the arithmetic is done in $\mbb{Z}_d$, $e^{i\theta}$ is a global phase factor depending on $(\mathbf{z},\mathbf{x})$, and $\mathbf{z}\cdot\mathbf{s}=\sum_{j=1}^nz_jb_j$ is the inner product in $\mbb{Z}_d^n$.
Thus we have
\begin{equation}
\begin{aligned}
A_{(0,0)}^{\otimes n}T_{(\mathbf{z},\mathbf{x})}^\dagger\ket{\Psi}&=e^{i\theta}\frac{1}{\sqrt{d^n}}\sum_{\mathbf{s}\in\mbb{Z}_d^n}\omega^{f(\mathbf{s}+\mathbf{x})-\mathbf{z}\cdot\mathbf{s}}\ket{-\mathbf{s}}\\
&=e^{i\theta}\frac{1}{\sqrt{d^n}}\sum_{\mathbf{s}\in\mbb{Z}_d^n}\omega^{f(-\mathbf{s}+\mathbf{x})+\mathbf{z}\cdot\mathbf{s}}\ket{\mathbf{s}}.
\end{aligned}
\end{equation}
By computing the overlap of $T_{(\mathbf{z},\mathbf{x})}^\dagger\ket{\Psi}$ and $A_{(0,0)}^{\otimes n}T_{(\mathbf{z},\mathbf{x})}^\dagger\ket{\Psi}$, we get
\begin{equation}
\begin{aligned}
W_{\Psi}(\mathbf{z},\mathbf{x})=\frac{1}{d^{2n}}\sum_{\mathbf{s}\in\mbb{Z}_d^n}\omega^{f(-\mathbf{s}+\mathbf{x})-f(\mathbf{s}+\mathbf{x})+2\mathbf{z}\cdot\mathbf{s}}.
\end{aligned}
\end{equation}
Now we can compute the Wigner negativity as
\begin{equation}
\begin{aligned}
\operatorname{sn}(\Psi)=&\frac{1}{2}\left(\sum_{\mathbf{z},\mathbf{x}\in\mbb{Z}_d^n}\abs{W_{\Psi}(\mathbf{z},\mathbf{x})}-1\right)\\
=&\frac{1}{2}\left(\frac{1}{d^{2n}}\sum_{\mathbf{z},\mathbf{x}\in\mbb{Z}_d^n}\abs{\sum_{\mathbf{s}\in\mbb{Z}_d^n}\omega^{f(-\mathbf{s}+\mathbf{x})-f(\mathbf{s}+\mathbf{x})+2\mathbf{z}\cdot\mathbf{s}}}-1\right).
\end{aligned}
\end{equation}

If $\Psi$ suffers from noise $\mc{E}_\lambda^{\otimes n}$, then
\begin{equation}\label{eq:noisy_qudit_HG_Wigner}
\begin{aligned}
W_{\mc{E}_\lambda^{\otimes n}(\Psi)}(\mathbf{z},\mathbf{x})=&\frac{1}{d^n}\Tr\left(\mc{E}_\lambda^{\otimes n}(A_{(\mathbf{z},\mathbf{x})})\ketbra{\Psi}{\Psi}\right)\\
=&\frac{1}{d^n}\bra{\Psi}T_{(\mathbf{z},\mathbf{x})}\left(\mc{E}_\lambda^{\otimes n}(A_{(0,0)})\right)^{\otimes n}T_{(\mathbf{z},\mathbf{x})}^\dagger\ket{\Psi}\\
=&\frac{1}{d^n}\sum_{I\subset[n]}(1-\lambda)^{n-\abs{I}}\lambda^{\abs{I}}\bra{\Psi}T_{(\mathbf{z},\mathbf{x})}\left(A_{(0,0)}^{\otimes n-\abs{I}}\otimes\frac{\mbb{I}_I}{d^{\abs{I}}}\right)T_{(\mathbf{z},\mathbf{x})}^\dagger\ket{\Psi}.
\end{aligned}
\end{equation}

\section{Example: Union Jack lattice}\label{app:UnionJack}

\begin{figure}[t]
\centering
\includegraphics[width=0.6\textwidth]{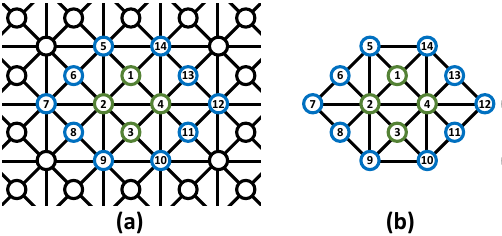}
\caption{Open cycles represent qubits. To obtain the hypergraph state, initialize each qubit in $\ket{+}$, and for every triangle in the lattice, apply a $\mathrm{CC}Z$ to its three vertices. (a) The hypergraph state $\Psi$ corresponds to the Union Jack lattice. (b) The 14-qubit hypergraph state $\Psi'$.}
\label{fig:UnionJack}
\end{figure}

Suppose $\Psi$ is the hypergraph state corresponding to the Union Jack lattice shown in Fig.~\ref{fig:UnionJack}(a). We calculate the RoM of the reduced density matrix of $\Psi$ on four qubits $1,2,3,4$. Since $1,2,3,4$ do not interact with other qubits other than $5,6,\cdots,14$, we have
\begin{equation}
\begin{aligned}
\Tr_{\overline{\{1,2,3,4\}}}(\Psi)=\Tr_{\overline{\{1,2,3,4\}}}(\Psi'),
\end{aligned}
\end{equation}
where $\Psi'$ is the hypergraph state shown in Fig.~\ref{fig:UnionJack}(b).
Numerical calculation shows that $\mc{R}\big(\Tr_{\overline{\{1,2,3,4\}}}(\Psi')\big)=1.0078125$.

\section{Example: $4$-complete hypergraph}\label{app:4-complete}

\begin{figure}[t]
\centering
\includegraphics[width=0.5\textwidth]{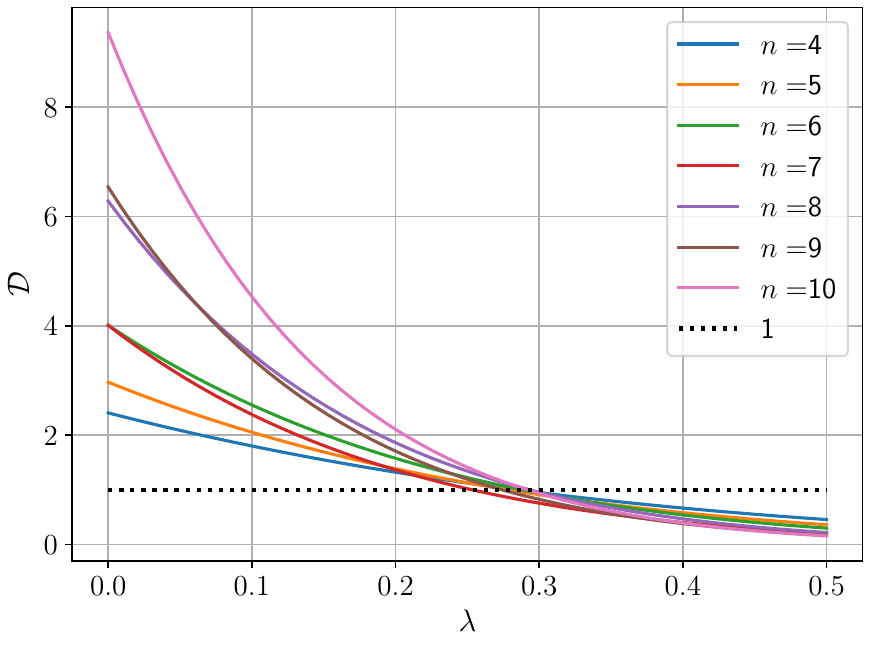}
\caption{Numerical results for $\mc{D}\left(\mc{E}_{\lambda}^{\otimes n}(\Psi_n)\right)$, where $\Psi_n$ is the $n$-qubit 4-complete hypergraph state.}
\label{fig:plot_D_4complete}
\end{figure}

Let $\Psi_n$ be the $n$-qubit 4-complete hypergraph state, which has all possible edges with degree 4 and no other edges.
The numerical results for $\mc{D}\left(\mc{E}_{\lambda}^{\otimes n}(\Psi_n)\right)$ are shown in Fig.~\ref{fig:plot_D_4complete}.

Using similar techniques as in Appendix~\ref{app:3-complete}, we can show that the reduced density matrix of $\Psi_n$ on three qubits tends to $\frac{1}{4}\sum_{l=0}^3\Phi_l$ as $n\rightarrow\infty$, where
\begin{equation}
\begin{aligned}
\ket{\Phi_0}&=\ket{+^3},\\
\ket{\Phi_1}&=\mathrm{CC}Z\ket{+^3},\\
\ket{\Phi_2}&=\mathrm{C}Z_{\{1,2\}}\mathrm{C}Z_{\{1,3\}}\mathrm{C}Z_{\{2,3\}}\ket{+^3},\\
\ket{\Phi_3}&=Z_{\{1\}}Z_{\{2\}}Z_{\{3\}}\mathrm{C}Z_{\{1,2\}}\mathrm{C}Z_{\{1,3\}}\mathrm{C}Z_{\{2,3\}}\mathrm{CC}Z\ket{+^3}.
\end{aligned}
\end{equation}
$\frac{1}{4}\sum_{l=0}^3\Phi_l$ has RoM $1.25$, thus $\{\Psi_n\}$ has $3$-local magic $1.25+o(1)$.

\section{Example: Qudit $\ket{\mathrm{C}^{n-1}Z}$}\label{app:qudit_CCCZ}

In this section we focus on the noise robustness and threshold of the $n$-qudit state $\ket{\mathrm{C}^{n-1}Z}=\frac{1}{\sqrt{d^n}}\sum_{\mathbf{s}\in\mbb{Z}_d^n}\omega^{\prod_jb_j}\ket{\mathbf{s}}$.

\subsection{Upper bound}\label{app:quditCCCZUB}

Similar to Lemma~\ref{lemma:HGstate_PartialTrace}, we have the following result:
\begin{lem}\label{lemma:HGstateQudit_PartialTrace}
For $d$ a prime number, let $\Psi$ be a $n$-qudit hypergraph state on $(\mbb{C}^d)^{\otimes n}$. For $I\subset [n]$ a subset of qudits, consider tracing out the qudits in $I$, we have
\begin{equation}
\Tr_{I}(\Psi)=\frac{1}{d^{\abs{I}}}\sum_{\mathbf{b}\in\mbb{Z}_d^{\abs{I}}}\Psi^{(I,\mathbf{b})},
\end{equation}
where $\Psi^{(I,\mathbf{b})}:=d^{\abs{I}}\bra{\mathbf{b}}_I\Psi\ket{\mathbf{b}}_I$ are $(n-\abs{I})$-qudit hypergraph states.
\end{lem}

We can obtain $\Psi^{(I,\mathbf{b})}$ from $\Psi$ by: For each qudit $i\in I$ which is labeled by $b_i$, remove this qudit and set the multiplicity of the edges containing this qudit to $b_i$ times (do arithmetic in $\mbb{Z}_d$) the original multiplicity.

\begin{lem}\label{lemma:qudit_CCCZ_R_UB}
For the $n$-qudit hypergraph sate $\rho_{n,\alpha}:=(\mathrm{C}^{n-1}Z)^\alpha\ketbra{+^n}{+^n}(\mathrm{C}^{n-1}Z^{\dagger})^\alpha$ which has no edges with degree $\le n-1$ and has an $n$-degree edge with multiplicity $\alpha\in\{1,\cdots,d-1\}$, we have
\begin{equation}
\mc{R}(\rho_{n,\alpha})\le1+4M_d(d-1)^n,
\end{equation}
where $M_d:=\sup_{\rho\in D(\mbb{C}^d)}\mc{R}(\rho)$.
\end{lem}
\begin{proof}
Denote $f(\x)=\prod_{j=1}^nx_j$ for $\x=(x_1,\cdots,x_n)\in\mbb{Z}_d^n$. Notice that 
\begin{equation}
\rho_{n,\alpha}=\frac{1}{d^n}\sum_{\x,\y\in\mbb{Z}_d^n}(\omega^\alpha)^{f(\x)-f(\y)}\ketbra{\x}{\y},
\end{equation}
most of whose entries are $d^{-n}$ for big $n$, thus $\rho_{n,a}$ is ``close" to the stabilizer state $\ketbra{+^n}{+^n}$.
Similar to the proof of Lemma~\ref{lem:CCCZ_magic_UB}, we can write
\begin{equation}\label{eq:qudit_CCCZ_decoposition}
\begin{aligned}
\rho_{n,\alpha}-\ketbra{+^n}{+^n}=&\frac{1}{d^n}\sum_{l=0}^{d-1}\sum_{f(\x)=l,f(\y)\neq l}\left((\omega^\alpha)^{l-f(\y)}\ketbra{\x}{\y}-\ketbra{\x}{\y}\right)\\
=&\frac{1}{d^n}\sum_{l=1}^{d-1}\sum_{f(\x)=l,f(\y)< l}\left((\omega^\alpha)^{l-f(\y)}\ketbra{\x}{\y}+(\omega^\alpha)^{f(\y)-l}\ketbra{\y}{\x}-\ketbra{\x}{\y}-\ketbra{\y}{\x}\right)\\
=&\frac{1}{d^n}\sum_{l=1}^{d-1}\sum_{f(\x)=l,f(\y)< l}\left(2\ketbra{B^{(\alpha l-\alpha f(\y))}_{\x\y}}{B^{(\alpha l-\alpha f(\y))}_{\x\y}}-2\ketbra{B^{(0)}_{\x\y}}{B^{(0)}_{\x\y}}\right),
\end{aligned}
\end{equation}
where $\ket{B^{(p)}_{\x\y}}=\frac{1}{\sqrt{2}}\left(\ket{\x}+\omega^{-p}\ket{\y}\right)$ for $p\in\mbb{Z}_d$ and $\x\neq\y$. 

We can use Clifford unitaries $X$ and $\operatorname{SUM}$ defined by
\begin{equation}
\begin{aligned}
X\ket{a}&=\ket{a+1~\operatorname{mod}~d},\\
\operatorname{SUM}\ket{a}\ket{b}&=\ket{a}\ket{a+b~\operatorname{mod}~d},
\end{aligned}
\end{equation}
to turn all but one qudit of $\ket{B^{(p)}_{\x\y}}$ into $\ket{0}$. (To achieve this, we first choose the qudit $j$ such that $x_j\neq y_j$, then apply $\operatorname{SUM}$ repeatedly to qudits $(j,j')$ for all $j'\neq j$, and finally apply $X$ repeatedly to qudit $j'$ for all $j'\neq j$. Now all except the $j$-th qudit becomes $\ket{0}$.) Thus we have $\mc{R}\big(B^{(p)}_{\x\y}\big)\le M_d$, where $M_d=\sup_{\rho\in D(\mbb{C}^d)}\mc{R}(\rho)$.

We have $\abs{f^{-1}(0)}=d^n-(d-1)^n$ and $\abs{f^{-1}(j)}=(d-1)^{n-1}$ for $j\neq0$ (which can be shown by induction on $n$). Thus we can count the terms in the summation of Eq.~\eqref{eq:qudit_CCCZ_decoposition} as 
\begin{equation}
\begin{aligned}
\sum_{l=1}^{d-1}\sum_{f(\x)=l,f(\y)< l}1=&\sum_{l=1}^{d-1}(d-1)^{n-1}\big(d^n-(d-1)^n+(l-1)(d-1)^{n-1}\big)\\
=&(d-1)^{n}\big(d^n-(d-1)^n+(d-1)^{n-1}(d-2)/2\big)<(d-1)^nd^n.
\end{aligned}
\end{equation}
Thus by Eq.~\eqref{eq:qudit_CCCZ_decoposition} and convexity of $\mc{R}$, we get
\begin{equation}
\begin{aligned}
\mc{R}(\rho_{n,\alpha})\le1+d^{-n}(d-1)^nd^n4M_d=1+4M_d(d-1)^n,
\end{aligned}
\end{equation}
which completes the proof.
\end{proof}

For $\rho_{n,1}=\ketbra{\mathrm{C}^{n-1}Z}{\mathrm{C}^{n-1}Z}$, we have
\begin{equation}
\begin{aligned}
\mc{E}_\lambda^{\otimes n}(\rho_{n,1})=\sum_{I\subset[n]}(1-\lambda)^{n-\abs{I}}\lambda^{\abs{I}}\Tr_{I}(\rho_{n,1})\otimes\mbb{I}_I.
\end{aligned}
\end{equation}
By permutation symmetry of $\rho_{n,1}$, we have
\begin{equation}
\begin{aligned}
\mc{R}\left(\mc{E}_\lambda^{\otimes n}(\rho_{n,1})\right)\le\sum_{k=0}^n\binom{n}{k}(1-\lambda)^{n-k}\lambda^k\mc{R}\left(\Tr_{[k]}(\rho_{n,1})\right).
\end{aligned}
\end{equation}
By Lemma~\ref{lemma:HGstateQudit_PartialTrace}, we have
\begin{equation}
\begin{aligned}
\Tr_{[k]}(\rho_{n,1})=\frac{1}{d^k}\sum_{\mathbf{b}\in\mbb{Z}_d^k}\rho_{n-k,\prod_{j=1}^kb_j},
\end{aligned}
\end{equation}
where $\mc{R}\left(\rho_{n-k,0}\right)=1$, and $\mc{R}\left(\rho_{n-k,\alpha}\right)$ with $\alpha\neq0$ can be upper bounded by Lemma~\ref{lemma:qudit_CCCZ_R_UB}. 
There are $d^{k}-(d-1)^{k}$ values of $\mathbf{b}\in\mbb{Z}_d^{k}$ such that $\prod_{j=1}^{k}b_j=0$, and $(d-1)^{k-1}$ values of $\mathbf{b}\in\mbb{Z}_d^{k}$ such that $\prod_{j=1}^{k}b_j=\alpha$ for each $\alpha\in\mbb{Z}_d-\{0\}$.
Thus
\begin{equation}\label{eq:qudit_CCCZ_UB_app}
\begin{aligned}
\mc{R}\left(\mc{E}_\lambda^{\otimes n}(\rho_{n,1})\right)\le&\sum_{k=0}^n\binom{n}{k}(1-\lambda)^{n-k}\lambda^k\frac{1}{d^k}\sum_{\mathbf{b}\in\mbb{Z}_d^k}\mc{R}\left(\rho_{n-k,\prod_{j=1}^kb_j}\right)\\
\le&\sum_{k=0}^n\binom{n}{k}(1-\lambda)^{n-k}\lambda^k\frac{1}{d^k}\Big((d^k-(d-1)^k)\cdot1+(d-1)^k(1+4M_d(d-1)^{n-k})\Big)\\
=&1+4M_d(d-1)^n\Big(1-\frac{d-1}{d}\lambda\Big)^n.
\end{aligned}
\end{equation}

When $\lambda>\frac{d(d-2)}{(d-1)^2}$, then $(d-1)\left(1-\frac{d-1}{d}\lambda\right)<1$, this upper bound converge to $1$ exponentially with regard to $n$.
For example, for qutrit $\rho_{n,1}$, the upper bound is $1+4M_32^n\left(1-2\lambda/3\right)^n$, and $\frac{d(d-2)}{(d-1)^2}=0.75$.

\subsection{Lower bound}\label{app:quditCCCZLB}

Consider the qutrit case ($d=3$) in this subsection. Let $S=\left\{4m-1\mid m\ge1\right\}=\{3,7,11,\cdots\}$. Consider the family of qutrit states $\big\{\ket{\mathrm{C}^{n-1}Z}\big\}_{n\in S}$. Note that all $n\in S$ are odd.

Denote $\Psi_n=\ketbra{\mathrm{C}^{n-1}Z}{\mathrm{C}^{n-1}Z}$. Take $\mathbf{1}=(1,1,\cdots,1)\in\mbb{Z}_3^n$, $\mathbf{0}=(0,0,\cdots,0)\in\mbb{Z}_3^n$, by Eq.~\eqref{eq:noisy_qudit_HG_Wigner} and the permutation symmetry of $\mathbf{1}$, $\mathbf{0}$, and $\ket{\mathrm{C}^{n-1}Z}$, we have
\begin{equation}
\begin{aligned}
W_{\mc{E}_\lambda^{\otimes n}(\Psi_n)}(\mathbf{1},\mathbf{0})=&\frac{1}{3^n}\Tr\left(\mc{E}_\lambda^{\otimes n}(A_{(\mathbf{1},\mathbf{0})})\ketbra{\mathrm{C}^{n-1}Z}{\mathrm{C}^{n-1}Z}\right)\\
=&\frac{1}{3^n}\sum_{k=0}^n\binom{n}{k}(1-\lambda)^{n-k}\lambda^{k}3^{-k}\bra{\mathrm{C}^{n-1}Z}T_{(\mathbf{1},\mathbf{0})}\left(A_{(0,0)}^{\otimes n-k}\otimes\mbb{I}_{3}^{\otimes k}\right)T_{(\mathbf{1},\mathbf{0})}^\dagger\ket{\mathrm{C}^{n-1}Z}.
\end{aligned}
\end{equation}
By Lemma~\ref{lem:qutrit_lemma} below, we know that 
\begin{equation}
\begin{aligned}
3^{2n}W_{\mc{E}_\lambda^{\otimes n}(\Psi_n)}(\mathbf{1},\mathbf{0})=&(1-\lambda)^n\frac{1}{2}\left(3-3^{\frac{n+1}{2}}\right)+\lambda^n3^{-n}3^n+\sum_{k=2}^{n-1}\binom{n}{k}\frac{1+(-1)^k}{2}(1-\lambda)^{n-k}\lambda^{k}3^{-k}\cdot3\cdot2^{k-1}\\
=&-\frac{1}{2}\cdot3^{\frac{n+1}{2}}(1-\lambda)^n+\lambda^n+\frac{3}{4}\left(1-\frac{\lambda}{3}\right)^n+\frac{3}{4}\left(1-\frac{5\lambda}{3}\right)^n\\
\le&-\frac{1}{2}\cdot3^{\frac{n+1}{2}}(1-\lambda)^n+1+\frac{3}{4}+\frac{3}{4}=-\frac{1}{2}\cdot3^{\frac{n+1}{2}}(1-\lambda)^n+\frac{5}{2}.
\end{aligned}
\end{equation}
Let $-\frac{1}{2}\cdot3^{\frac{n+1}{2}}(1-\lambda_n)^n+\frac{5}{2}=0$, then $W_{\mc{E}_\lambda^{\otimes n}(\Psi_n)}(\mathbf{1},\mathbf{0})<0$ for $\lambda<\lambda_n$. We have $\lambda_n=1-5^{\frac{1}{n}}3^{-\frac{n+1}{2n}}\rightarrow1-3^{-1/2}\gtrsim0.42$ as $n\rightarrow\infty$. 

\begin{lem}\label{lem:qutrit_lemma}
We have
\begin{equation}
\begin{aligned}
3^n\bra{\mathrm{C}^{n-1}Z}T_{(\mathbf{1},\mathbf{0})}\left(A_{(0,0)}^{\otimes n-k}\otimes\mbb{I}_{3}^{\otimes k}\right)T_{(\mathbf{1},\mathbf{0})}^\dagger\ket{\mathrm{C}^{n-1}Z}=
\begin{cases}
\frac{1}{2}\left(3-3^{\frac{n+1}{2}}\right), &\text{ when }k=0,\\
0, &\text{ when }k=1,3,\cdots,n-2,\\
3\cdot2^{k-1}, &\text{ when }k=2,4,\cdots,n-1,\\
3^n, &\text{ when }k=n.\\
\end{cases}
\end{aligned}
\end{equation}
\end{lem}

\begin{proof}
Note that
\begin{equation}
\begin{aligned}
T_{(\mathbf{1},\mathbf{0})}^\dagger\ket{\mathrm{C}^{n-1}Z}=e^{i\theta}\frac{1}{\sqrt{3^n}}\sum_{\mathbf{s}\in\mbb{Z}_3^n}\omega^{\prod_{j=1}^ns_j-\sum_{j=1}^ns_j}\ket{\mathbf{s}},
\end{aligned}
\end{equation}
where $e^{i\theta}$ is a global phase factor, and
\begin{equation}
\begin{aligned}
\left(A_{(0,0)}^{\otimes n-k}\otimes\mbb{I}_{3}^{\otimes k}\right)T_{(\mathbf{1},\mathbf{0})}^\dagger\ket{\mathrm{C}^{n-1}Z}=&e^{i\theta}\frac{1}{\sqrt{3^n}}\sum_{\mathbf{s}\in\mbb{Z}_3^n}\omega^{\prod_{j=1}^ns_j-\sum_{j=1}^ns_j}\ket{-s_1,\cdots,-s_{n-k},s_{n-k+1},\cdots,s_{n}}\\
=&e^{i\theta}\frac{1}{\sqrt{3^n}}\sum_{\mathbf{s}\in\mbb{Z}_3^n}\omega^{(-1)^{n-k}\prod_{j=1}^ns_j+\sum_{j=1}^{n-k}s_j-\sum_{j=n-k+1}^ns_j}\ket{\mathbf{s}}.
\end{aligned}
\end{equation}
Thus we have
\begin{equation}\label{eq:qutrit_target}
\bra{\mathrm{C}^{n-1}Z}T_{(\mathbf{1},\mathbf{0})}\left(A_{(0,0)}^{\otimes n-k}\otimes\mbb{I}_{3}^{\otimes k}\right)T_{(\mathbf{1},\mathbf{0})}^\dagger\ket{\mathrm{C}^{n-1}Z}=\frac{1}{3^n}\sum_{\mathbf{s}\in\mbb{Z}_3^n}\omega^{\left((-1)^{n-k}-1\right)\prod_{j=1}^ns_j+2\sum_{j=1}^{n-k}s_j}=:\frac{1}{3^n}\sum_{\mathbf{s}\in\mbb{Z}_3^n}\omega^{g(\mathbf{s})},
\end{equation}
where we define function $g:\mbb{Z}_3^n\rightarrow \mbb{Z}_3$ by
\begin{equation}
g(\mathbf{s})=\left((-1)^{n-k}-1\right)\prod_{j=1}^ns_j+2\sum_{j=1}^{n-k}s_j.
\end{equation}
Since $g(-\mathbf{s})=-g(\mathbf{s})$, we know that the number of $\mathbf{s}$'s such that $g(\mathbf{s})=1,2$ are equal, that is, $\abs{g^{-1}(1)}=\abs{g^{-1}(2)}$. Thus
\begin{equation}
\sum_{\mathbf{s}\in\mbb{Z}_3^n}\omega^{g(\mathbf{s})}=\abs{g^{-1}(0)}-\frac{1}{2}\left(3^n-\abs{g^{-1}(0)}\right)=\frac{3}{2}\abs{g^{-1}(0)}-\frac{3^n}{2}.
\end{equation}
The following of our proof is focused on calculating $\abs{g^{-1}(0)}$.

For $\mathbf{s}\in\mbb{Z}_3^n$ and $j\in\mbb{Z}_3$, denote $\operatorname{wt}_j(\mathbf{s})$ be the number of $j$'s in the string $\mathbf{s}$.

(\textbf{Case 1.}) When $k=0$, we have
\begin{equation}
\begin{aligned}
g(\mathbf{s})=-\Big(2\prod_{j=1}^ns_j+\sum_{j=1}^{n}s_j\Big),
\end{aligned}
\end{equation}
thus $g(\mathbf{s})=0$ iff $\prod_{j=1}^ns_j=\sum_{j=1}^{n}s_j$. Note that
\begin{equation}
\prod_{j=1}^ns_j=
\begin{cases}
0,&\operatorname{wt}_0(\mathbf{s})>0,\\
1,&\operatorname{wt}_0(\mathbf{s})=0\text{ and }\operatorname{wt}_2(\mathbf{s})\text{ even},\\
2,&\operatorname{wt}_0(\mathbf{s})=0\text{ and }\operatorname{wt}_2(\mathbf{s})\text{ odd},\\
\end{cases}
\end{equation}
and for $l\in\mbb{Z}_3$,
\begin{equation}
\sum_{j=1}^{n}s_j=l,\text{ when }\operatorname{wt}_1(\mathbf{s})-\operatorname{wt}_2(\mathbf{s})\equiv l,
\end{equation}
where $\equiv$ denotes equivalent $\operatorname{mod}3$. We use $w_l$ to represent $\operatorname{wt}_l(\mathbf{s})$, which satisfies $w_0+w_1+w_2=n$. 

The number of $\mathbf{s}$ such that $\prod_{j=1}^ns_j=\sum_{j=1}^{n}s_j=0$ equals to
\begin{equation}\label{eq:qutrit_case1_0}
\sum_{w_0=1}^n\binom{n}{w_0}\sum_{w_1=0}^{n-w_0}\binom{n-w_0}{w_1}\delta_{w_1-w_2\equiv0}=\sum_{w_0=1}^n\binom{n}{w_0}\sum_{w_1=0}^{n-w_0}\binom{n-w_0}{w_1}\frac{(\omega^{2w_1+w_0-n}-\omega)(\omega^{2w_1+w_0-n}-\omega^2)}{(1-\omega)(1-\omega^2)}.
\end{equation}
The number of $\mathbf{s}$ such that $\prod_{j=1}^ns_j=\sum_{j=1}^{n}s_j=1$ equals to
\begin{equation}\label{eq:qutrit_case1_1}
\sum_{w_2=0}^{n}\binom{n}{w_2}\delta_{w_2\text{ even}}\delta_{w_1-w_2\equiv1}=\sum_{w_2=0}^{n}\binom{n}{w_2}\frac{1+(-1)^{w_2}}{2}\frac{(\omega^{n-2w_2}-1)(\omega^{n-2w_2}-\omega^2)}{(\omega-1)(\omega-\omega^2)}.
\end{equation}
The number of $\mathbf{s}$ such that $\prod_{j=1}^ns_j=\sum_{j=1}^{n}s_j=2$ equals to
\begin{equation}\label{eq:qutrit_case1_2}
\sum_{w_2=0}^{n}\binom{n}{w_2}\delta_{w_2\text{ odd}}\delta_{w_1-w_2\equiv2}=\sum_{w_2=0}^{n}\binom{n}{w_2}\frac{1-(-1)^{w_2}}{2}\frac{(\omega^{n-2w_2}-1)(\omega^{n-2w_2}-\omega)}{(\omega^2-1)(\omega^2-\omega)}.
\end{equation}
Summing Eq.~\eqref{eq:qutrit_case1_0}\eqref{eq:qutrit_case1_1}\eqref{eq:qutrit_case1_2}, we get
\begin{equation}
\begin{aligned}
\sum_{\mathbf{s}\in\mbb{Z}_3^n}\omega^{g(\mathbf{s})}=\frac{3}{2}\abs{g^{-1}(0)}-\frac{3^n}{2}=\frac{3}{2}+\frac{1}{4}3^{\frac{n+1}{2}}\left(ie^{-\frac{n\pi}{2}i}-ie^{\frac{n\pi}{2}i}\right)=\frac{1}{2}\left(3-3^\frac{n+1}{2}\right),
\end{aligned}
\end{equation}
where the last equality holds because we take $n=4m-1$.

(\textbf{Case 2.}) When $k=1,3,\cdots,n-2$, we have $g(\mathbf{s})=2\sum_{j=1}^{n-k}s_j$.
In this case the number of $\mathbf{s}$'s such that $g(\mathbf{s})=0,1,2$ are equal, that is, $\abs{g^{-1}(0)}=\abs{g^{-1}(1)}=\abs{g^{-1}(2)}=3^{n-1}$ (which can be shown by induction).

(\textbf{Case 3.}) When $k=2,4,\cdots,n-1$, we have
\begin{equation}
\begin{aligned}
g(\mathbf{s})=-\Big(2\prod_{j=1}^ns_j+\sum_{j=1}^{n-k}s_j\Big),
\end{aligned}
\end{equation}
thus $g(\mathbf{s})=0$ iff $\prod_{j=1}^ns_j=\sum_{j=1}^{n-k}s_j$. We use $w_l$ to represent $\operatorname{wt}_l(\mathbf{s})$, which satisfies $w_0+w_1+w_2=n$. Denote $\mathbf{s}'\in\mbb{Z}_3^{n-k}$ be the vector containing the first $n-k$ entries of $\mathbf{s}$, we use $w_l'$ to represent $\operatorname{wt}_l(\mathbf{s}')$, which satisfies $w_0'+w_1'+w_2'=n-k$.

The number of $\mathbf{s}$ such that $\prod_{j=1}^ns_j=\sum_{j=1}^{n-k}s_j=0$ equals to
\begin{equation}\label{eq:qutrit_case3_0}
\begin{aligned}
\sum_{w_0=1}^n\sum_{w_0'=0}^{w_0}\binom{n-k}{w_0'}\binom{k}{w_0-w_0'}2^{k-w_0+w_0'}\sum_{w_1'=0}^{n-k-w_0'}\binom{n-k-w_0'}{w_1'}\delta_{w_1'-w_2'\equiv0}
=3^{n-1}-2^{k}\sum_{w_1'=0}^{n-k}\binom{n-k}{w_1'}\delta_{w_1'-w_2'\equiv0}.
\end{aligned}
\end{equation}
The number of $\mathbf{s}$ such that $\prod_{j=1}^ns_j=\sum_{j=1}^{n-k}s_j=1$ equals to
\begin{equation}\label{eq:qutrit_case3_1}
\sum_{w_2=0}^{n}\sum_{w_2'=0}^{w_2}\binom{n-k}{w_2'}\binom{k}{w_2-w_2'}\delta_{w_2\text{ even}}\delta_{w_1-w_2\equiv1}=\sum_{w_2'=0}^{n-k}\sum_{w_2=w_2'}^{w_2'+k}\binom{n-k}{w_2'}\binom{k}{w_2-w_2'}\delta_{w_2\text{ even}}\delta_{w_1-w_2\equiv1}.
\end{equation}
The number of $\mathbf{s}$ such that $\prod_{j=1}^ns_j=\sum_{j=1}^{n-k}s_j=2$ equals to
\begin{equation}\label{eq:qutrit_case3_2}
\sum_{w_2=0}^{n}\sum_{w_2'=0}^{w_2}\binom{n-k}{w_2'}\binom{k}{w_2-w_2'}\delta_{w_2\text{ odd}}\delta_{w_1-w_2\equiv2}=\sum_{w_2'=0}^{n-k}\sum_{w_2=w_2'}^{w_2'+k}\binom{n-k}{w_2'}\binom{k}{w_2-w_2'}\delta_{w_2\text{ odd}}\delta_{w_1-w_2\equiv2}.
\end{equation}
Summing Eq.~\eqref{eq:qutrit_case3_0}\eqref{eq:qutrit_case3_1}\eqref{eq:qutrit_case3_2}, we get
\begin{equation}
\begin{aligned}
\sum_{\mathbf{s}\in\mbb{Z}_3^n}\omega^{g(\mathbf{s})}=\frac{3}{2}\abs{g^{-1}(0)}-\frac{3^n}{2}=-(-1)^n(-1)^k3\cdot2^{k-1}=3\cdot2^{k-1},
\end{aligned}
\end{equation}
where the last equality holds because $n$ odd $k$ even.

(\textbf{Case 4.}) When $k=n$, we have $g(\mathbf{s})=0$ for all $\mathbf{s}$, thus $\abs{g^{-1}(0)}=3^n$.

\end{proof}

\begin{figure}[t]
\centering
\includegraphics[width=0.9\textwidth]{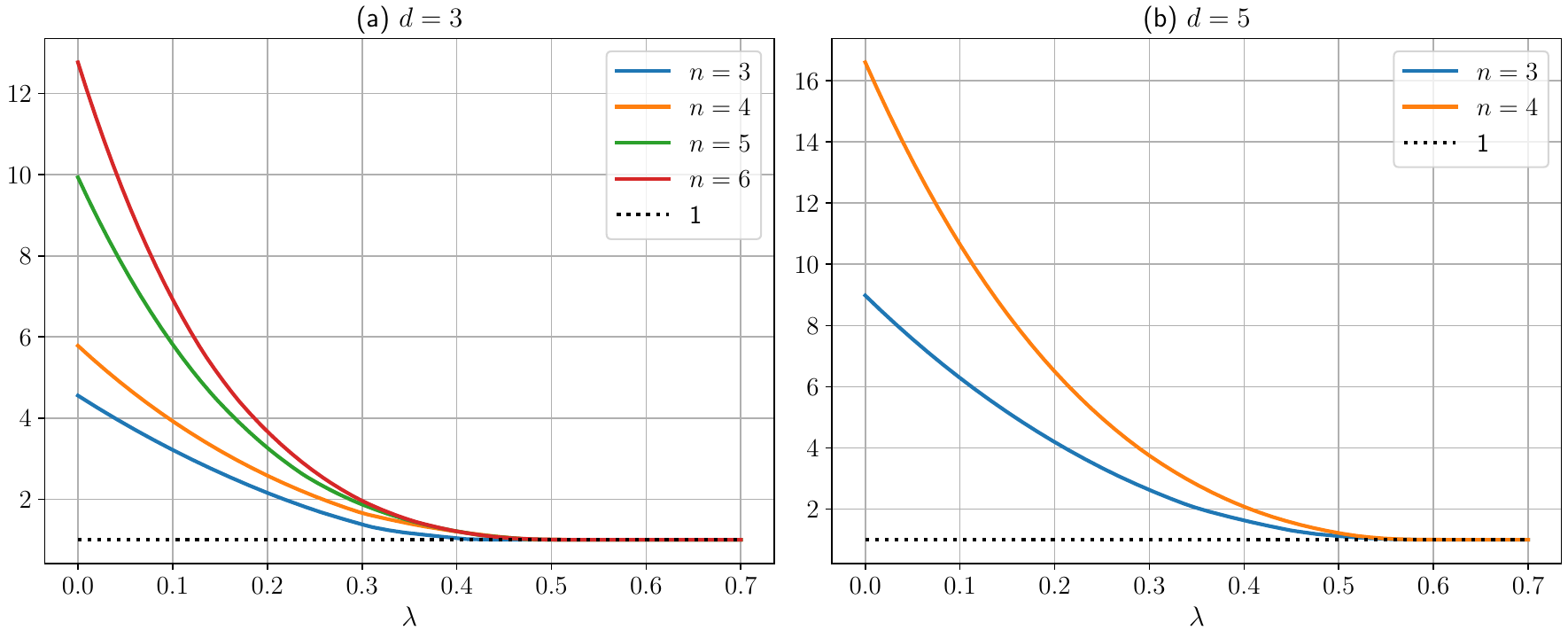}
\caption{Numerical values of $1+2\operatorname{sn}\left(\mc{E}_\lambda^{\otimes n}(\ketbra{\mathrm{C}^{n-1}Z}{\mathrm{C}^{n-1}Z})\right)$ with local dimension $d=3,5$.}
\label{fig:plot_N_qudit_sym_d3_d5}
\end{figure}

\subsection{Local magic}

Note that in Eq.~\eqref{eq:qudit_CCCZ_UB_app} we have shown that
\begin{equation}
\begin{aligned}
\mc{R}\left(\Tr_{[k]}(\rho_{n,1})\right)\le&\frac{1}{d^k}\Big(d^k+(d-1)^k4M_d(d-1)^{n-k}\Big)\\
=&1+4M_dd^{n-k}\big(\frac{d-1}{d}\big)^n.
\end{aligned}
\end{equation}
Take $k=n-3$, we have
\begin{equation}
\mc{R}\left(\Tr_{[n-3]}(\rho_{n,1})\right)\le1+4M_dd^{3}\Big(\frac{d-1}{d}\Big)^{n}=1+\mc{O}\Big(\Big(\frac{d-1}{d}\Big)^{n}\Big).
\end{equation}

\section{Fragility of edges with degrees near $n$}\label{app:Fragility_of_bigbigEdges}

\begin{lem}\label{lem:defferenceUB_by_wt}
For arbitrary $n$-qubit hypergraph states $\Psi_1$ and $\Psi_2$ with characteristic function $f_1$ and $f_2$, we have
\begin{equation}
\abs{\mc{R}(\Psi_1)-\mc{R}(\Psi_2)}<4\operatorname{wt}\left(f_1+f_2\right).
\end{equation}
\end{lem}
\begin{proof}
Note that
\begin{equation}
\begin{aligned}
\Psi_1-\Psi_2&=2^{-n}\sum_{x,y\in\mbb{Z}_2^n}\left[(-1)^{f_1(x)+f_1(y)}-(-1)^{f_2(x)+f_2(y)}\right]\ketbra{\x}{\y}\\
&=2^{-n}\sum_{x,y\in\mbb{Z}_2^n}(-1)^{f_1(x)+f_1(y)}\left[1-(-1)^{(f_1+f_2)(x)+(f_1+f_2)(y)}\right]\ketbra{\x}{\y}\\
&=2^{1-n}\sum_{\substack{(f_1+f_2)(x)=1\\(f_1+f_2)(y)=0}}(-1)^{f_1(x)+f_1(y)}\big(\ketbra{\x}{\y}+\ketbra{y}{x}\big)\\
&=2^{1-n}\sum_{\substack{(f_1+f_2)(x)=1\\(f_1+f_2)(y)=0}}(-1)^{f_1(x)+f_1(y)}\big(\ketbra{X^+_{xy}}{X^+_{xy}}-\ketbra{X^-_{xy}}{X^-_{xy}}\big),
\end{aligned}
\end{equation}
where $\ket{X^\pm_{xy}}=\frac{1}{\sqrt{2}}(\ket{x}\pm\ket{y})$ are stabilizer states. By convexity of $\mc{R}$ we have
\begin{equation}
\mc{R}(\Psi_1)\le\mc{R}(\Psi_2)+4\frac{\operatorname{wt}(f_1+f_2)\big(2^n-\operatorname{wt}(f_1+f_2)\big)}{2^n}<\mc{R}(\Psi_2)+4\operatorname{wt}(f_1+f_2).
\end{equation}
By symmetric we also have $\mc{R}(\Psi_2)<\mc{R}(\Psi_1)+4\operatorname{wt}(f_1+f_2)$.
\end{proof}

\begin{lem}\label{lem:noisy_add_one_big_edge}
For an $n$-qubit hypergraph state $\Psi$, consider adding an $m$-edge $e\subset[n]$ to $\Psi$ and obtaining a new $n$-qubit hypergraph state $\Phi$. Then we have
\begin{equation}
\begin{aligned}
\mc{R}\left(\mc{E}_{\lambda}^{\otimes n}(\Phi)\right)\le\mc{R}\left(\mc{E}_{\lambda}^{\otimes n}(\Psi)\right)+4\cdot 2^{n-m}(1-\lambda/2)^n.
\end{aligned}
\end{equation}
\end{lem}
\begin{proof}
From the proof of Theorem~\ref{thm:HighDegreeEdgeFragile_app}, we know
\begin{equation}
\begin{aligned}
\mc{E}_{\lambda}^{\otimes n}(\Phi)-\mc{E}_{\lambda}^{\otimes n}(\Psi)=&\sum_{I\subset[n]}\Bigg[(1-\lambda)^{n-\abs{I}}\lambda^{\abs{I}}\frac{1}{2^{\abs{I}}}\sum_{\s|_{I\cap e}=1^{\abs{I\cap e}}}\left[\Phi^{(I,\s)}-\Psi^{(I,\s)}\right]\otimes\frac{\mbb{I}_I}{2^{\abs{I}}}\Bigg].
\end{aligned}
\end{equation}
Denote $f_1$ and $f_2$ be the characteristic function of the $n-\abs{I}$ qubit states $\Phi^{(I,\s)}$ and $\Psi^{(I,\s)}$, respectively. We have $(f_1+f_2)(x)=\prod_{i\in e-I}x_i$, thus $\operatorname{wt}(f_1+f_2)=2^{(n-\abs{I})-\abs{e-I}}=2^{n-m-\abs{I-e}}$. By Lemma~\ref{lem:defferenceUB_by_wt}, we know that $\Phi^{(I,\s)}-\Psi^{(I,\s)}$ can be written as a mixture of stabilizer states with total coefficients $4\cdot2^{n-m-\abs{I-e}}$, thus
\begin{equation}
\begin{aligned}
\mc{R}\left(\mc{E}_{\lambda}^{\otimes n}(\Phi)\right)\le\mc{R}\left(\mc{E}_{\lambda}^{\otimes n}(\Psi)\right)+\sum_{I\subset[n]}\Bigg[(1-\lambda)^{n-\abs{I}}\lambda^{\abs{I}}\frac{1}{2^{\abs{I}}}\sum_{\s|_{I\cap e}=1^{\abs{I\cap e}}}4\cdot2^{n-m-\abs{I-e}}\Bigg].
\end{aligned}
\end{equation}
The number of $I\subset{[n]}$ with $\abs{I}=k$, such that $\abs{I\cap e}=a$ (which means $\abs{I-e}=k-a$), equals to $\binom{m}{a}\binom{n-m}{k-a}$. When $\abs{I\cap e}=a$, the number of $s$ such that $\s|_{I\cap e}=1^{\abs{I\cap e}}$ equals to $2^{k-a}$. Note that for $p<q$ we have $\binom{p}{q}=0$. Thus we have
\begin{equation}
\begin{aligned}
\mc{R}\left(\mc{E}_{\lambda}^{\otimes n}(\Phi)\right)\le&\mc{R}\left(\mc{E}_{\lambda}^{\otimes n}(\Psi)\right)+4\cdot2^{n-m}\sum_{k=0}^n\Bigg[(1-\lambda)^{n-k}\lambda^k2^{-k}\sum_{a=0}^k\binom{m}{a}\binom{n-m}{k-a}2^{k-a}2^{-(k-a)}\Bigg]\\
=&\mc{R}\left(\mc{E}_{\lambda}^{\otimes n}(\Psi)\right)+4\cdot2^{n-m}\sum_{k=0}^n\binom{n}{k}(1-\lambda)^{n-k}\lambda^k2^{-k}\\
=&\mc{R}\left(\mc{E}_{\lambda}^{\otimes n}(\Psi)\right)+4\cdot2^{n-m}\left(1-\lambda/2\right)^n.
\end{aligned}
\end{equation}
\end{proof}

\begin{prop}\label{prop:HighDegreeEdgeFragile2}
Let $\Psi$ be a $n$-qubit hypergraph state. If $\Phi$ is obtained by adding arbitrary many edges with degree $\ge n-c$ from $\Psi$, where $c$ is a constant, then we have
\begin{equation}
\mc{R}\left(\mc{E}_{\lambda}^{\otimes n}\left(\Phi\right)\right)\le\mc{R}\left(\mc{E}_{\lambda}^{\otimes n}\left(\Psi\right)\right)+\operatorname{poly}(n)(1-\lambda/2)^n.
\end{equation}
\end{prop}
\begin{proof}
Consider adding the edges one by one, and denote the sequence of hypergraphs by $\Psi=\Phi_0\overset{e_1}{\rightarrow}\Phi_1\overset{e_2}{\rightarrow}\Phi_2\overset{e_3}{\rightarrow}\cdots\overset{e_{K}}{\rightarrow}\Phi_K=\Phi$. Since $c$ is a constant, we have $K=\operatorname{poly(n)}$. Apply Lemma~\ref{lem:noisy_add_one_big_edge} sequentially, we have
\begin{equation}
\begin{aligned}
\mc{R}\left(\mc{E}_{\lambda}^{\otimes n}\left(\Phi\right)\right)\le&\mc{R}\left(\mc{E}_{\lambda}^{\otimes n}\left(\Phi_{K-1}\right)\right)+4\cdot2^c\left(1-\lambda/2\right)^{n}\\
\le&\mc{R}\left(\mc{E}_{\lambda}^{\otimes n}\left(\Phi_{K-2}\right)\right)+2\cdot4\cdot2^c\left(1-\lambda/2\right)^{n}\\
\le&\cdots\cdots\\
\le&\mc{R}\left(\mc{E}_{\lambda}^{\otimes n}\left(\Psi\right)\right)+K\cdot4\cdot2^c\left(1-\lambda/2\right)^{n}\\
=&\mc{R}\left(\mc{E}_{\lambda}^{\otimes n}\left(\Psi\right)\right)+\operatorname{poly}(n)\left(1-\lambda/2\right)^{n}.
\end{aligned}
\end{equation}
    
\end{proof}

\end{document}